\documentclass[12pt]{article}

\usepackage[english]{babel}

\usepackage[a4paper,top=2.5cm,bottom=2.5cm,left=3cm,right=3cm,marginparwidth=1.75cm]{geometry}

\usepackage{amsmath}
\usepackage{graphicx}
\usepackage[colorlinks=true, allcolors=blue]{hyperref}
\usepackage{array}
\usepackage{float}
\usepackage{multirow}
\usepackage{amsthm}
\usepackage{amssymb}
\usepackage{graphicx}
\usepackage{tablefootnote}
\usepackage{setspace}
\usepackage[round]{natbib}
\usepackage{enumerate}
\usepackage{comment}
\usepackage{booktabs}
\usepackage{pgfplots}
\pgfplotsset{compat=1.18}
\theoremstyle{plain}
\newtheorem{lem}{\protect\lemmaname}
\newtheorem{prop}{\protect\propositionname}
\newtheorem*{theorem}{Theorem}
\theoremstyle{definition}
\newtheorem*{defn*}{\protect\definitionname}
\newtheorem{defn}{\protect\definitionname}
\newtheorem{cor}{\protect\corollaryname}
\newtheorem{hypothesis}{Hypothesis}

\usetikzlibrary{positioning}

\providecommand{\lemmaname}{Lemma}
\providecommand{\corollaryname}{Corollary}
\providecommand{\definitionname}{Definition}
\providecommand{\propositionname}{Proposition}

\title{Does joint liability reduce cheating in contests with agency problems?
Theory and experimental evidence\thanks{Funding from the International Olympic Committee's Anti-Doping Research Fund is gratefully acknowledged.}}
\author{Qin Wu\thanks{School of Economics, Finance and Marketing, RMIT University, Australia (Email:
qin.wu@rmit.edu.au).}, Ralph-C Bayer\thanks{School of Economics and Public Policy, University of Adelaide, Australia (Email:
ralph.bayer@adelaide.edu.au).}}

\begin{document}

\maketitle

\begin{abstract}
Contest participants often have strong incentives to engage in cheating. Sanctions serve as a common deterrent against such conduct. Often, other agents on the contestant's team (e.g., a coach of an athlete) or a company (a manager of an R\&D engineer) have a vested interest in outcomes and can influence the cheating decision. An agency problem arises when only the contestant faces the penalties for cheating.
Our theoretical framework examines joint liability, i.e., shifting some responsibility from the contestant to the other agent, as a solution to this agency problem. Equilibrium analysis shows that extending liability reduces cheating if fines are harsh. Less intuitively, when fines are lenient, a shift in liability can lead to an increase in equilibrium cheating rates. Experimental tests confirm that joint liability is effective in reducing cheating if fines are high. However, the predicted detrimental effect of joint liability for low fines does not occur.


\end{abstract}
\newpage
\section{Introduction}

In high-stake contests, participants face strong incentives to cheat. To maintain fairness, contest organizers often implement rules and laws that impose penalties for cheaters. When an individual contestant has full control over the decision to cheat, the classical crime and punishment frameworks provide a basis for understanding their trade-offs between the expected benefits and costs of cheating. However, the situation becomes far more complex when other agents, such as managers, advisors, or support personnel, can influence the decision of a contestant to cheat or not.

The issue of non-contesting team members influencing cheating decisions arises across various domains. In sports, doping decisions often involve not only athletes but also their coaches, medical personnel, and team officials. For example, state-organized doping programs in Russia highlight the critical role external agents can play in enabling misconduct. Similarly, in the 2012 doping scandal involving the Essendon Football Club in the Australian Football League, coaching staff and sports scientists implemented a supplements program that included substances later deemed illegal. Sometimes the athletes might not even know that they are given performance enhancing drugs by their coaches as inthe case of British sprinter Bernice Wilson. Outside professional sports, high school and college athletics have seen cases in which trainers pressured young athletes to use performance-enhancing drugs to secure scholarships or endorsements.

Beyond sports, corporate scandals reveal similar patterns. In the Volkswagen emissions scandal, engineers tasked with developing diesel engines designed illegal defeat devices to cheat on emissions tests as part of their efforts to compete in the race for technological innovation. However, these actions were heavily influenced by the senior executives who directed and approved the scheme. Similarly, in the Enron scandal, financial officers and traders manipulated accounting practices to inflate earnings, while high-ranking executives orchestrated the fraudulent strategies, including concealing billions of dollars in debt.

In many of these cases, the contestants bear the primary blame, while the agents who influence or facilitate their cheating often avoid direct consequences. The involvement of influencing agents, who benefit from a contestant’s performance without directly participating in the contest, creates an agency problem. This issue is exacerbated when liability for detected cheating falls solely on the contestant. An intuitive remedy for this agency problem involves the implementation of joint liability, which reallocates a portion of the liability for sanctions to non-participating agents involved in the decision-making process.

This paper investigates whether joint liability necessarily yields a reduction in cheating. In a basic contest model with agency frictions, we show that transferring some liability from the contestant to the agent (which we call "the manager") can discourage cheating. Somewhat surprisingly, this is only true if fines are harsh. In contrast, when fines are low, joint liability can even result in an increased equilibrium cheating rate. Introducing joint liability typically discourages managers from cheating, while simultaneously increasing the contestant's incentives to cheat. The magnitude of fines, together with complex general equilibrium effects, determines which effect dominates and whether joint liability reduces or increases the overall cheating rate in equilibrium. 

The key lesson from our theoretical investigations is that a principal who aims at reducing cheating should consider transferring some liability to non-contestants only if fines are sufficiently severe. However, this insight relies on the assumption that agents and contestants adhere to equilibrium behavior. Given the complexity of strategic interaction and  subtle general equilibrium effects, it is unclear whether actual behavior is aligned with these theoretical predictions. In order to assess this, we take our model to the laboratory and investigate if the mechanism outlined by theory guides actual behavior and leads to the hypothesized comparative statics of cheating frequencies. 

In a two-by-two design, we manipulate the sanctioning rule (individual liability versus joint liability) and the severity of sanctions (high fines versus low fines). The experimental results only align in part with the theoretical predictions. Consistent with theory, we find that joint liability significantly reduces cheating when fines are high. However, contrary to the predicted backfiring effect, where joint liability is expected to increase cheating under low fines, we observe no such increase. To further investigate the robustness of this finding, we introduce two additional treatments with even lower fine levels, where the backfiring effect should be even stronger. Behavior in the additional treatments suggests that the absence of the backfiring effect is robust.


The absence of the backfiring effect is good news for policy makers. When considering to shift some liability to non-contestants, they do not have to perform the daunting task of figuring out if fines are above the critical level for joint liability to decrease cheating rates. In the worst case, according to our experimental results, introducing joint liability will leave cheating rates unchanged, with the upside of potentially reducing them.

\section{Literature Review}

A substantial body of theoretical research explores cheating behavior in contests, particularly focusing on the mechanisms that can deter such behavior. In individual contest settings, \citet{curry2009deterrence} use a general contest success function to study the impact of prize structures and differential monitoring on deterrence costs. They find that differential monitoring between winners and losers, re-awarding prizes forfeit by caught cheaters, and awarding second prizes lowers the monitoring costs required to deter cheating. \citet{gilpatric2011cheating} studies cheating behavior in a symmetric rank-order tournament, where two players make effort and cheating decisions. The study compares correlated audits (where either all contestants are audited or none) with uncorrelated audits (random individual audits), concluding that correlated audits are more effective in deterring cheating. \citet{ryvkin2013contests} examines doping within a contest environment with multiple contestants, incorporating random and independent testing after players make cheating decisions. The study shows that the minimal enforcement parameters required to prevent doping in equilibrium can be non-monotonic with respect to the number of players. Moreover, it highlights the crucial role of penalties in large tournaments, as they contribute to reducing the minimal testing probability as the number of players increases.

While these studies provide insights into individual cheating behavior in competitive settings, matters become more complex when multiple individuals in a team can cheat. This is even true without the complication of competing with other teams. Research has shown that team settings significantly influence cheating behavior. For example, \citet{sutter2009deception} compares individuals and teams in an experiment where participants could deceive others by telling the truth in a misleading way. The study finds that teams are more likely to engage in such deceptive strategies than individuals, suggesting that team settings lower the moral barriers to unethical behavior. One contributing factor is the diffusion of responsibility in teams. \citet{conrads2013lying} find that individuals in team settings feel less personally accountable for dishonest actions, which encourages more cheating. Exposure to unethical behavior through communication is another factor that leads to increased cheating. \citet{kocher2018lie} demonstrate that individuals are more likely to engage in unethical practices when they observe or discuss such behaviors within their group. Moreover, when rewards for cheating are shared among team members, individuals tend to cheat more by over-reporting their performance, as individual moral concerns are diluted when the benefits are collective \citep{wiltermuth2011cheating}.

Despite a general aversion to cheating \citep{gneezy2005deception,lundquist2009aversion,abeler2019preferences,hurkens2009would}, individuals often justify their actions by discounting the immorality of cheating when their dishonest behavior benefits others \citep{schweitzer2002stretching}. This moral disengagement is facilitated by team dynamics, where in-group favoritism leads individuals to cheat more for in-group members at the expense of out-group members \citep{cadsby2016group}. \citet{danilov2013dark} provide experimental evidence supporting this notion, showing that strong group affiliations can lead to more unethical behaviors, such as financial advisers deliberately recommending low-quality products to customers.


The influence of agents (managers) on cheating decisions introduces an additional layer of complexity. This is particulalry important in corporate settings where individual actions can significantly impact organizational outcomes. \citet{friebel2004abuse} investigate how abuse of authority within organizations can lead to undesirable outcomes such as cheating. \citet{pierce2008ethical} examine ethical spillovers in firms, focusing on vehicle emissions testing facilities. Using data from over three million emissions tests, they find that inspectors adjust their behavior to align with the ethical norms of the organizations they work for, indicating that firms can influence employee behavior through formal norms and incentives. 

The issue of agency problems affecting cheating decisions in competitive environments has received limited attention from economists. The most closely related paper is \citet{crocker2005corporate}, who develop a principal-agent model to study cheating behavior in the context of corporate tax evasion. In their model, shareholders (the principals) employ a Chief Financial Officer (CFO, the agent) who possesses private information about allowable deductions and makes decisions about tax declarations on behalf of the firm. The CFO may engage in illegal tax evasion by claiming unauthorized deductions, which, if detected, incur penalties. Shareholders incentivize the CFO to reduce the company's tax burden by designing a compensation contract that ties the CFO's salary inversely to the effective tax rate. The authors examine the efficacy of penalties levied on either the shareholders or the CFO alone and compare these to a joint liability regime. They find that penalties imposed directly on the CFO are more effective in reducing evasion than those imposed on shareholders. 

While their study provides insights into the internal agency dynamics that drive cheating decisions, it is limited to a single-firm context where the shareholders influence the CFO’s decisions only indirectly through the compensation contract. The model does not account for the external competitive pressures, which are often critical in real-world scenarios.

Our study contributes to the literature by examining cheating behavior in a competitive environment. Our work differs from previous studies in two key aspects. First, both the manager and the contestant can directly influence the final cheating decision, which introduces a principal-agent component to the model. Second, our model introduces competition with other teams, introducing strategic considerations that impact the decision to cheat and the effectiveness of joint liability as a deterrent. Moreover, unlike \citet{crocker2005corporate}, who focus on the internal dynamics of a single firm, our competitive framework captures how the actions of external competitors interact with internal agency dynamics and influence cheating behavior. 



The paper proceeds as follows. In the next Section, we develop our theoretical model and derive equilibrium predictions. We focus on how equilibrium cheating rates depend on the size of the fine and the distribution of liability. Section \ref{sec:Experimental-Design4} details our experimental design. Section \ref{sec:Results4} presents the main results of the experiments. Section \ref{robustness} reports on two additional treatments that provide a robustness check on the result that joint liability does not backfire even if fines are low. The final Section offers a brief conclusion.

\section{\label{sec:Theoretical-Model4}Theoretical model}

In our model there are two teams, $1$ and $2$. Each team consists of two players, a contestant $c$ and a manager $m$. Hence, the set of players is $I=\{c1,m1,c2,m2\}$. The contestants exert effort in a share-prize contest against each other. The managers do not exert effort, but still receive a share of the prize their contestants secure in the contest. Teams can decide to cheat or not. The team's cheating decision is jointly determined by the contestant and the manager. For given efforts and the cheating decision of the other team, cheating increases the share of the total prize a team receives.  

In stage 1, all four players choose whether they would like their team to cheat or not. Denote player $ri$'s cheating intention as $d_{ri}\in\left\{ 0,1\right\}$, where $r$ denotes the role and $i$ the team of the player. A value of one indicates that an agent would prefer her team to cheat. 
Once cheating intentions are chosen, nature randomly (with equal probability) implements the intention of one team member as the final cheating decision $d_{i}\in\{0,1\}$ for the team. This random dictator mechanism is used
as a reduced-form description of an otherwise more complex negotiation
process. We chose this simplified setup with the experimental implementation in mind. Simplicity is important to keep confusion and the resulting noise in behavior to a minimum. Our simple reduced-form mechanism captures the salient features of bargaining by ensuring that the likelihood of a team cheating increases with each team member's intention to cheat. Moreover, if both team members agree on an action, then the agreed action is implemented with certainty.

In stage two, after observing the implemented cheating decisions $\{d_1,d_2\}$ of both teams, the two contestants $c1$ and $c2$ simultaneously exert efforts $e_{1}$ and $e_2$. The effort and the cheating decisions determine the outcome of the contest. The prize received by team $i$ is denoted by $S_{i}(e_{i},e_{j},d_{i},d_{j})$. The contestant receives fraction
$r \in (0,1)$ of the prize. The remainder goes to the manager.

While the assumption that contestants observe cheating decisions may seem restrictive, it is not unrealistic in many real-world settings. In sports, for example, athletes often show noticeable physical changes associated with cheating, and their opponents can recognize unusually rapid performance improvements without direct evidence of doping. Similarly, in corporate environments, unexpectedly high reported profits or efficiency can raise suspicions of unethical practices, alerting competitors to potential misconduct.

These observable signs of cheating justify our use of complete information as a simplifying assumption. It allows us to derive closed-form solutions for optimal effort decisions, which simplifies the analysis considerably. Moreover, in the parameter regions where teams cheat or do not cheat with certainty, the equilibria remain the same regardless of perfect or imperfect information.

In contrast to \citet{berentsen2002economics} who assumes maximum effort and thus avoids the strategic considerations arising from effort responses to cheating, our model allows contestants to choose their effort levels conditional on cheating decisions. This added complexity, combined with the involvement of a second party that influences the cheating decision, makes the complete information assumption essential to maintain tractability.


Finally, with probability $p$, independent audits uncover cheating 
if it has occurred. If a team is caught cheating, the team has to pay a fine $f$. 
The proportion $\eta \in [0,1]$ of the fine is borne by the contestant, while the manager pays the share $(1-\eta)$. 
Note that the parameter $\eta$ captures the liability regime. The case of $\eta=1$ implies  
individual liability, where the contestant is liable alone. For $\eta<1$ we have joint
liability. Lower values of $\eta$ correspond to regimes that put higher liability on 
the managers. An interesting special case is $\eta=r$, where fines and
contest revenues are shared at the same rate.

Given this setup, contestant $ci$'s expected payoff is 
\begin{equation}
\pi_{ci}(e_{i},e_{j};d_{i},d_{j})=rS_{i}(e_{i},e_{j},d_{i},d_{j})-e_{i}-d_{i}pf\eta,\label{eq:payoff}
\end{equation}
where $rS_{i}$ is team $i$'s prize share that goes to the contestant, $d_{i}pf$
is the expected fine, $\eta$ is the proportion of the fine levied on the contestant, and effort cost are $e_i$. The manager $mi$ receives the remainder of the price and has to bear his share of the expected find but does not have any effort cost:
\[
\pi_{mi}(e_{i},e_{j};d_{i},d_{j})=(1-r)S_{i}(e_{i},e_{j},d_{i},d_{j})-(1-\eta)d_{i}pf.
\]

Without loss of generality, we normalize the total prize to one, which implies that
the fine parameter $f<1$ can be interpreted as fraction of the total prize that a team has to pay if caught cheating.
Similarly, effort costs are measured as a fraction of the total prize.
The prize share $S_{i}$ is determined by a standard Tullock contest
function, where efforts are augmented by an effectiveness factor $\theta$,
which depends on the implemented cheating decision $d_{i}$. Cheating
makes effort more effective in delivering a share of the prize (i.e.
$\theta$$\left(d_{i}=1\right)>\theta\left(d_{i}=0\right)$). For
simplicity, we assume 

\begin{equation}
\theta(d_{i})=\begin{cases}
1 & \textrm{if }d_{i}=0\\
\delta & \textrm{if }d_{i}=1,
\end{cases}\label{eq:delta}
\end{equation}

\noindent where $\delta>1$ captures the effectiveness of cheating. Under these assumptions team $i$'s share becomes:
\[
S_{i}(e_{i},e_{j},d_{i},d_{j})=\begin{cases}
\frac{\theta(d_{i})e_{i}}{\theta(d_{i})e_{i}+\theta(d_{j})e_{j}} & \textrm{if }e_{i}+e_{j}>0\\
\frac{\theta(d_{i})}{\theta(d_{i})+\theta(d_{j})} & \textrm{if }e_{i}+e_{j}=0.
\end{cases}
\]

Note that there is an agency problem which cannot be removed, even if
the liability is shared in the same way as the revenues from the contest.
This agency friction arises from the contestant having to exert effort, which causes
costs that are not shared with the manager.

\subsection{Optimal efforts}

In what follows, we will use backward induction to characterize the Subgame-Perfect Nash equilibria of the game. We start solving the model by first determining the equilibrium efforts conditional on the implemented cheating decisions. The Lemma below describes the equilibrium efforts.
\begin{lem}
\noindent \label{lemma1}In equilibrium, effort levels are identical for both contestants: 
\begin{equation}
e_{i}^{*}=e_{j}^{*}=\frac{r\theta(d_{i})\theta(d_{j})}{\left(\theta(d_{i})+\theta(d_{j})\right){}^{2}}.
\end{equation}
\end{lem}

\begin{proof}
\noindent The first-order condition for contestant $ci$ yields:
\begin{equation}
\frac{r e_j \theta(d_{i})\theta(d_{j})}{\left(e_i \theta(d_{i})+e_j \theta(d_{j})\right){}^{2}}=1,
\label{eq:breff}
\end{equation}
which implies $e_{i}^{*}=e_{j}^{*}$, as the fraction is invariant to a swap of the indices. Solving gives the stated result.
\end{proof}

Optimal efforts show an interesting dependence on the implemented
cheating decisions. The contestants exert the highest efforts when
the contest is even. So if both either cheat or both abstain from
cheating, then nobody has an advantage in the contest, and equilibrium
efforts are those of a standard Tullock contest. Efforts 
are lower in an unequal contest, where one contestant has the advantage of
being the sole cheat:
\begin{equation}
e_{i}^{*}\left(d_{i},d_{j}\right)=\begin{cases}
\frac{r}{4} & \textrm{if }d_{i}=d_{j}\\
\frac{r\delta}{\left(1+\delta\right)^{2}} & \textrm{if }d_{i}\neq d_{j}.
\end{cases}\label{eq:estar}
\end{equation}

This has interesting implications for the
strategic link between efforts and cheating.
Cheating and efforts are complements if the competitor cheats and
substitutes if the competitor does not cheat. 

We further note that the value of $\eta$ does not impact
the equilibrium efforts, since it does not change the marginal cost or benefits of effort exertion. 
This implies that Lemma \ref{lemma1} applies under
all liability schemes and shifting liability is effort neutral, as long as it does not change the cheating decisions.

\subsection{The contestant's cheating decision}

In what follows, we determine the optimal cheating decision of a contestant
conditional on the believed probability that the other team cheats. Denote the optimal 
cheating decision of contestant $i$ given the believed cheating probability of the other
team $\mu_j$ as $\sigma^{*}_{ci}(\mu_j)$.
In order to determine the optimal cheating decision of the contestant,
we calculate the continuation payoffs depending on the implemented
cheating decisions for contestant $ci$ by substituting the optimal efforts from (\ref{eq:estar})
into the payoff function in (\ref{eq:payoff}):
\[
\pi_{ci}^{*}(d_{i},d_{j})=r\left(\frac{\theta(d_{i})}{\theta(d_{i})+\theta(d_{j})}\right)^{2}-d_{i}pf\eta.
\]

This allows us to calculate the expected payoffs for contestant $ci$ if her team cheats $\mathbb{E}\pi_{ci}^*(1,\mu_j)$ and if her team stays clean $\mathbb{E}\pi_{ci}^*(0,\mu_j)$:
\begin{align}
    \mathbb{E}\pi_{ci}^*(1,\mu_j)&=r\left(\frac{\mu_j}{4}+(1-\mu_j)\left(\frac{\delta}{1+\delta}\right)^2 \right)-p f \eta \\
    \mathbb{E}\pi_{ci}^*(0,\mu_j)&=r\left(\frac{\mu_j}{(1+\delta)^2}+\frac{1-\mu_j}{4}\right)
\end{align}

Observe that $\mathbb{E}\pi_{ci}^*(1,\mu_j)$ decreases in $\mu_j$, while $\mathbb{E}\pi_{ci}^*(0,\mu_j)$ increases in the other team's cheating probability.\footnote{To see this note that $1/(1+\delta)^2<1/4<(\delta/(1+\delta))^2$, as $\delta>1$.}  This implies that the payoff difference between cheating and not cheating decreases in the other team's cheating probability.

\begin{defn}
Define the belief $\mu_j$ that equalizes the expected payoffs for cheating and not cheating as $\hat{\mu}$
\begin{equation*}
    \hat{\mu}:=\left\{\mu_j \in \mathbb{R} : \mathbb{E}\pi_{ci}^*(0,\mu_j) = \mathbb{E}\pi_{ci}^*(1,\mu_j) \right\}
\end{equation*}
\begin{equation}
    \hat{\mu}=\frac{r\left( \left(\frac{\delta}{1+\delta}\right)^{2}-\frac{1}{4}\right)-pf\eta}{r \left( \frac{1+\delta^{2}}{\left(1+\delta\right)^{2}}-\frac{1}{2}\right)}.\label{eq:crith}
\end{equation}
\end{defn}

The contestant prefers her team to cheat whenever the other team's cheating probability is below $\hat{\mu}$, is indifferent if they are equal, and prefers her team not to cheat if $\mu_j$ exceeds $\hat{\mu}$. This, together with the observation that increasing the probability that the own team cheats is optimal whenever the contestant prefers her team to cheat, allows us to pin down the best responses of the contestants. 

\begin{prop}
\label{Propconoptimal}Contestant $ci$'s best response to the opposing team's cheating probability is given by:

\begin{eqnarray}
\sigma_{ci}^{*}(\mu_j) & = & \begin{cases}
0 & \text{if \ensuremath{\mu_{j}>\hat{\mu}}}\\
\in[0,1] & \text{if \ensuremath{\mu_{j}=\hat{\mu}}}\\
1 & \text{if \ensuremath{\mu_{j}<\hat{\mu}}}
\end{cases}.\label{eq:contcheat}
\end{eqnarray}
\end{prop}

Note that $\hat{\mu}$ is not necessarily between zero and one. For the case that $\hat{\mu}<0$ the contestants have a dominant strategy not to cheat. If, on the other hand, $\hat{\mu}>1$ then the contestant's dominant strategy is to cheat. 

Furthermore, $\hat{\mu}$ decreases with the expected fine levied on the contestant $pf\eta$. This implies that, in line with intuition, the incentive to cheat for the contestant increases if some of the liability is shifted to the manager. 

\subsection{The manager's cheating decision}

Let us now consider the managers' cheating decisions. Manager $mi$'s
continuation payoff anticipating optimal efforts is given by:
\begin{equation}
  \pi_{mi}^{*}(d_{i},d_{j})=(1-r)\left(\frac{\theta(d_{i})}{\theta(d_{i})+\theta(d_{j})}\right)-d_{i}pf(1-\eta).  
\end{equation}

Calculating the payoff for $mi$ if his team cheats $\mathbb{E}\pi_{mi}^*(1,\mu_j)$ and if his team plays by the rules $\mathbb{E}\pi_{mi}^*(0,\mu_j)$ yields:
\begin{align}
    \mathbb{E}\pi_{mi}^*(1,\mu_j)&=(1-r)\left(\frac{\mu_j}{2}+\frac{\delta(1-\mu_j)}{1+\delta}\right)-p f(1- \eta) \\
    \mathbb{E}\pi_{ai}^*(0,\mu_j)&=(1-r)\left(\frac{\mu_j}{1+\delta}+\frac{1-\mu_j}{2}\right)
\end{align}
The difference between the expected payoff from cheating and not cheating is:
\begin{equation*}
    \mathbb{E}\pi_{mi}^*(1,\mu_j)-\mathbb{E}\pi_{mi}^*(0,\mu_j)=(1-r)\left(\frac{1}{2}-\frac{1}{1+\delta}\right)-p f(1- \eta),
\end{equation*}
which does not depend on $\mu_j.$ Hence, depending on the sign of the payoff difference, the manager has a dominant strategy to cheat, not to cheat or is indifferent. The difference increases with the fraction of the fine that is borne by the contestant $\eta$. 

\begin{defn}
Define the contestant's liability $\eta$ that equalizes the manager's expected payoffs for cheating and not cheating as $\hat{\eta}$
\begin{equation*}
    \hat{\eta}:=\left\{\eta \in \mathbb{R} : \mathbb{E}\pi_{mi}^*(0,\mu_j) = \mathbb{E}\pi_{mi}^*(1,\mu_j) \right\}
\end{equation*}
\begin{equation}
    \hat{\eta}=1-\frac{(\delta-1)(1-r)}{(1+\delta)2fp}.\label{eq:crit_eta}
\end{equation}
\end{defn}
The higher the cheating efficiency ($\delta$), the fraction of the prize the manager receives ($1-r$) and the lower the expected fine ($pf$) the higher is the minimum fraction of the fine that must be levied on the manager ($1-\hat{\eta}$) to deter cheating. 

\begin{prop}
\label{Propmanjo} The manager's optimal choice does not depend on the other team's doping probability and is given by:
\begin{equation}
\sigma_{mi}^{*}=\begin{cases}
0 & \text{if \ensuremath{\eta<\hat{\eta}}}\\
\in[0,1] & \text{if \ensuremath{\eta=\hat{\eta}}}\\
1 & \text{if \ensuremath{\eta>\hat{\eta}}}
\end{cases}
\end{equation}
\end{prop}

\subsection{Equilibria}
We have established that the managers' optimal decisions depend only on the parameters, while the contestants' optimal cheating decisions typically depend on the cheating probability of the opposing team. Hence, constructing an equilibrium begins with fixing the managers' decisions. Let us begin by looking for equilibria, where both managers cheat. This requires a parameter constellation, where $\eta \in [\hat{\eta},1]$. In this region, the penalty share $1-\eta$ the manager bears must be sufficiently low. Now observe that for $\hat{\mu}\geq 1$, we have a unique pure-strategy equilibrium where all four actors cheat. Such an equilibrium occurs when the expected fine $pf$ is negligible. 

Next, suppose that the environment is less favorable for cheating such that $\hat{\mu}\in[1/2,1).$ In this case we obtain two asymmetric equilibria where, on top of the two managers, one of the contestants cheats. This is the case as the cheating probability of the team with the cheating contestant increases to unity, so that the other contestant's best response is not to cheat. This in turn ensures that the cheating contestant does not want to deviate, since the opposing team's cheating probability is $1/2$ and does not exceed $\hat{\mu}$. The average cheating probability is $3/4$. Note that in this case a mixed-strategy equilibrium, where the contestants randomize, also exists. Both contestants cheat with probability $2\hat{\mu}-1$ in order to make the other contestant indifferent. The average cheating probability in this equilibrium becomes $\hat{\mu}$. 

In an even less favorable environment for cheating with $\hat{\mu}\in[0,1/2]$, in which $\eta \geq \hat{\eta}$ still holds, we have an equilibrium, where both contestants choose not to cheat. The average cheating probability is $1/2$. Denoting the probability of a randomly selected team cheating 
\[\bar{p}=\left(\sigma_{m1}^{*}+\sigma_{a1}^{*}+\sigma_{m2}^{*}+\sigma_{a2}^{*}\right)/4, \]
we can summarize the findings discussed above in the following Proposition.

\begin{prop}
\label{propeqm} For $\eta \in [\hat{\eta},1]$ and $i,j \in\{1,2\}, i \neq j$ we obtain the following equilibria, where both managers cheat:
\begin{enumerate}
\item $\hat{\mu}\geq1$:

$\left(\sigma_{mi}^{*},\sigma_{ci}^{*};\sigma_{mj}^{*},\sigma_{cj}^{*}\right)=\left(1,1;1,1\right)$
; $\bar{p}=1$

\item $\hat{\mu}\in\left[1/2,1\right]$:

$\left(\sigma_{mi}^{*},\sigma_{ci}^{*};\sigma_{mj}^{*},\sigma_{cj}^{*}\right)=\left(1,1;1,0\right),\left(1,0;1,1\right)$;
$\bar{p}=3/4$;

$\left(\sigma_{mi}^{*},\sigma_{ci}^{*};\sigma_{mj}^{*},\sigma_{cj}^{*}\right)=\left(1,2\hat{\mu}-1;1,2\hat{\mu}-1\right)$; $\bar{p}=\hat{\mu}$  

\item $\hat{\mu}\leq 1/2$:

$\left(\sigma_{mi}^{*},\sigma_{ci}^{*};\sigma_{mj}^{*},\sigma_{cj}^{*}\right)=\left(1,0;1,0\right)$;
$\bar{p}=1/2$.
\end{enumerate}
\end{prop}

\begin{proof}
The construction of the equilibria is straight-forward and was explained above. If remains to show that parameter configurations exist, where $\eta \in [\hat{\eta},1]$ and (a) $\hat{\mu}\geq1$, (b) $\hat{\mu}\in\left[1/2,1\right]$, and (c) $\hat{\mu}\leq 1/2$ hold, such that all equilibra exist. It is sufficient to show that $\eta \in [\hat{\eta},1]$ is compatible with $\hat{\mu}=1$ and with $\hat{\mu}=1/2$. First, let $\eta=1$ such that $\eta \in [\hat{\eta},1]$ is satisfied. substituting into (\ref{eq:crith}), $\hat{\mu}=1$ requires 
\begin{equation*}
    pf=\frac{r(\delta+1)(\delta+3)}{4(1-\delta)^2},    
\end{equation*}
which is feasible. Next let $\eta=\hat{\eta}$, such that $\eta \in [\hat{\eta},1]$ holds and substitute (\ref{eq:crit_eta}) into (\ref{eq:crith}) and set equal to $1/2$, which requires:
\begin{equation*}
    pf=\frac{\delta-1}{2(\delta+1)},
\end{equation*}
which is also feasible.
\end{proof}

Next, we look for equilibria in which managers have a dominant strategy not to cheat. This is the case whenever $\eta \in [0,\hat{\eta}]$. The construction of equilibria is very similar as above. Given that managers now do not cheat, the maximum probability for a team to cheat is $1/2$. Hence, there is a pure-strategy equilibrium where both contestants cheat if the parameters are such that contestants find it preferable to cheat if the other team's cheating probability is equal to $1/2$, which is the case whenever, $\hat{\mu} \geq 1/2$. Whenever $\hat{\mu} \in [0,1/2]$ we obtain two asymmetric equilbria, where one of the two contestants cheats while the other does not. The third equilibrium in this parameter region is symmetric, where both contestants mix. Finally, there is also an equilibrium where all four agents do not cheat. For this we require $\hat{\mu}\leq 0$. In the following Proposition, we state these findings more formally.

\begin{prop} \label{propeqm2}
 For $\eta \in [0,\hat{\eta}]$ and $i,j \in\{1,2\}, i \neq j$ we obtain the following equilibria, where both managers do not cheat.
\begin{enumerate}
\item $\hat{\mu}\geq 1/2$:

$\left(\sigma_{mi}^{*},\sigma_{ci}^{*};\sigma_{mj}^{*},\sigma_{cj}^{*}\right)=\left(0,1;0,1\right)$
; $\bar{p}=1/2$

\item $\hat{\mu}\in\left[0,1/2\right]$:

$\left(\sigma_{m1}^{*},\sigma_{ci}^{*};\sigma_{mj}^{*},\sigma_{cj}^{*}\right)=\left(0,1;0,0\right),\left(0,0;0,0\right)$;
$\bar{p}=1/4$;

$\left(\sigma_{mi}^{*},\sigma_{ci}^{*};\sigma_{mj}^{*},\sigma_{cj}^{*}\right)=\left(0,2\hat{\mu};0,2\hat{\mu}\right)$;$\bar{p}=\hat{\mu}$  

\item $\hat{\mu}\leq 0$:

$\left(\sigma_{mi}^{*},\sigma_{ci}^{*};\sigma_{mj}^{*},\sigma_{cj}^{*}\right)=\left(0,0;0,0\right)$;
$\bar{p}=0$.
\end{enumerate}
\end{prop}

\begin{proof}
In order to show that all these equilibria exist we have to show that $\eta \in [0,\hat{\eta}]$ is compatible with both $\hat{\mu}=0$ and $\hat{\mu}=1/2$. Let $\eta=\hat{\eta}$, such that $\eta \in [\hat{\eta},1]$ holds and substitute (\ref{eq:crit_eta}) into (\ref{eq:crith}) and set equal to $1/2$ and to $0$, which requires:

\begin{gather*}
    pf=\frac{\delta-1}{2(1+\delta)} \ \textrm{and} \\
    pf=\frac{(\delta-1)(2-r+\delta(2+r))}{4(1+\delta)^2},
\end{gather*}

which are both feasible.
\end{proof}

For reasons of brevity, we omit the non-generic equilibria for $\eta=\hat{\eta}$ where managers randomize. Moreover, these equilibria do not yield qualitatively different cheating probabilities. The average cheating probability for any of these equilibria lies between those resulting from the equilibria, where both managers cheat and those where both managers do not cheat.\footnote{Note that for every possible $\hat{\mu}$ for which $\eta=\hat{\eta}$ there exists a continuum of equilibria where any combination of cheating probabilites of the managers is possible and the contestants best-respond to the resulting aggregate cheating probability of the opposing team.}

\subsection{The impact of liability and expected fines on cheating}
The distribution of liability for detected cheating is captured by the parameter $\eta$. Under individual liability the contestant bears the entire negative consequences of a team being caught cheating. In such a setting with $\eta=1$ there is nothing that discourages managers from cheating, and in equilibrium the probability of cheating is at least 50 percent. A natural candidate to remedy this is to shift some of the liability for incurred fines to the managers, as this reduces their cheating incentives. Although reducing $\eta$ reduces the incentive for cheating of the managers, it makes cheating more attractive to contestants. Hence, it is not a priori clear that such a shift will reduce the over-all cheating probability. The relationship between the equilibrium cheating probability and the distribution of fines is complex. This is due to a misalignment of the preferences of the contestants and the managers that originates from the contestants competing in the contest, while the managers stay on the sidelines.

For joint liability to reduce the cheating probability, we require that on the one hand, the resulting share of the fine for the manager (i.e. $1-\eta$) incentivzes the managers to refrain from cheating, while on the other hand the now reduced fine share $\eta$ on the contestants is still sufficient to deter them from cheating. Joint liability will backfire and increase the aggregate cheating likelihood if the shift of the fine share incentivizes the contestants to cheat without deterring the managers. It turns out that the crucial factor that determines whether joint liability can reduce cheating is the expected fine $pf$. If the expected fine is sufficiently high (compared to the efficiency of cheating $\delta$), then joint liability has the potential to reduce cheating without the risk of backfiring. If the expected fine is low, then joint liability in equilibrium will never reduce cheating but might backfire and result in a higher average cheating probability compared to individual liability. 

Before we can formalize the described intuition in a Theorem, we need to define the maximum and minimum expected cheating probability for a certain parameter constellation. This is helpful as we have multiple equilibria for two parameter regions ($\eta\geq\hat{\eta} \land \hat{\mu} \in [0,1/2]$ and $\eta\leq\hat{\eta} \land \hat{\mu} \in [1/2,1]$). Denote the maximum average cheating probability in any equilibrium given the liability share $\eta$ and expected fine $pf$ for given other parameters $r$ and $\delta$ as $\bar{p}_{h}(\eta,pf)$ and the minimum cheating probability as $\bar{p}_{l}(\eta,pf)$. 

\begin{theorem} \label{theorem}
For any given parameter combination of $\delta>1$, $r \in(0,1)$ and $pf$ for which 
\begin{equation*}
    pf>\frac{r(\delta-1)}{2(1+\delta)} 
\end{equation*}
liability $\eta$ has the following impact on expected equilibrium cheating probabilities:

\begin{enumerate}

    \item If $pf \in \left(  \frac{r(\delta-1)}{2(1+\delta)},\frac{\delta-1}{2(1+\delta)}\right)$, then
    
      $\bar{p}_{l}(\eta,pf)\geq \bar{p}_{h}(1,pf) \ \forall \eta \in[0,1]$ and $\exists \eta \in [0,1]$ such that $\bar{p}_{l}(\eta,pf)> \bar{p}_{h}(1,pf)$.
     
     \item If $pf= \frac{\delta-1}{2(1+\delta)}$, then 
     
      $\bar{p}_{l}(\eta,pf)= \bar{p}_{h}(\eta,pf)=1/2 \ \forall \eta \in[0,1]$.
     
      \item If $pf> \frac{\delta-1}{2(1+\delta)}$, then
    
      $\bar{p}_{h}(\eta,pf)\leq \bar{p}_{l}(1,pf) \ \forall \eta \in[0,1]$ and $\exists \eta \in [0,1]$ such that $\bar{p}_{h}(\eta,pf)< \bar{p}_{l}(1,pf)$.
      
\end{enumerate}
\end{theorem}

\begin{proof}
See appendix.
\end{proof}

\begin{figure}[h]
\begin{tikzpicture}
	\begin{axis}[
		width=.9\textwidth,
		xlabel=$\eta$,
		ylabel=$pf$,
		ymin=0.11, ymax=0.225,
		xmin=0.39, xmax=1.03,
	]

	\addplot [
	    black, 
	    domain=0.5:0.9,
	  ] 
	    {1/(18*(1-x)};

		\addplot [
	    black, 
	    domain=0.3:1,
	  ] 
	    {1/(9*x)};

		\addplot [
	    black, 
	    domain=5/8:1,
	  ] 
	    {5/(54*x)};

		\addplot [
	    black, 
	    domain=0.3:7/10,
	  ] 
	    {7/(54*x)};
	    
	    \addplot [
	    black,
	    dotted
	  ]
	    {1/6};
	    
	    \node[above] at (0.85,0.18) {$\bar{p}=1/2$};
	    
	   \node[above] at (0.5,0.14) {$\bar{p}=1/2$};
	   
	   \node[above] at (0.67,0.12) {$\bar{p}=1$};
	   
	   \node[above] at (0.67,0.21) {$\bar{p}=0$};
	   
	   \node[pin=below left:{$\bar{p} \in (0,1/2)$}] at (0.59,0.2) {};
	   
	   \node[pin=above right:{$\bar{p} \in (1/2,1)$}] at (0.8,0.13) {};
	   
	   \node[above] at (0.44,1/6) {$pf=1/6$};
	   
	   \filldraw[black] (2/3,11/72) circle (2pt) node[anchor=east]{\textit{Jo\_L}};
	   
	   \filldraw[black] (1,11/72) circle (2pt) node[anchor=east]{\textit{Ind\_L}};
	   
	   \filldraw[black] (2/3,13/72) circle (2pt);
	   
	   \node[] at (0.67,0.184) {\textit{Jo\_H}};
	   
	   \filldraw[black] (1,13/72) circle (2pt) node[anchor=east]{\textit{Ind\_H}};
	   
	   \draw[black] (2/3,1/9) circle (2pt) node[anchor=south]{\textit{Jo\_40}};
	   
	   \draw[black] (1,1/9) circle (2pt) node[anchor=south]{\textit{Ind\_40}};
	 
	 \end{axis}
\end{tikzpicture}
\caption{Equilibrium Cheating Probabilities for $\delta=2$ and $r=2/3$} \label{fig: cheatingprob}
\end{figure}
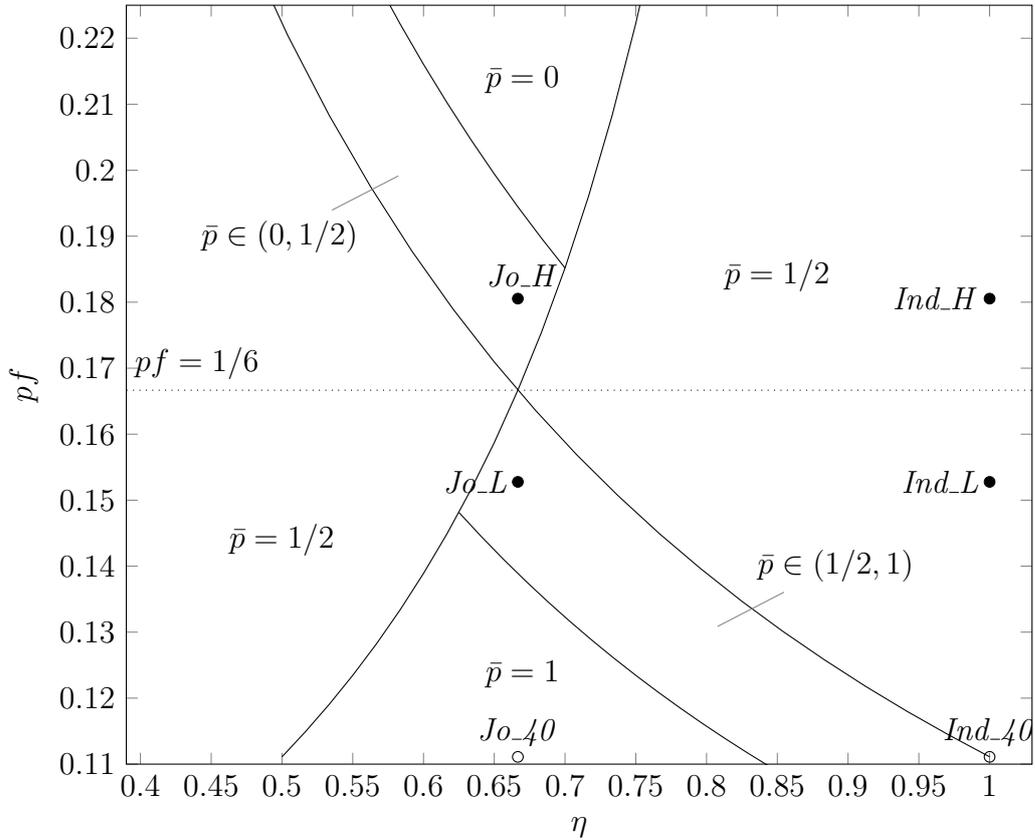

Our theorem establishes that the size of the fine determines if joint liability can help reduce cheating or if it might backfire and increase the cheating likelihood. It covers the most interesting case where the fine is at least high enough to deter the contestants from cheating if they bear the whole fine.\footnote{Note that for lower expected fines backfiring can still be a problem as decreasing $\eta$ can switch the equilibrium from both managers and one contestant to all four agents cheating.} This is the case for $pf>r(\delta-1)/(2(1+\delta))$. Then under a modest fine joint liability will never reduce the rate of cheating but will for some liability share $1-\eta$ actually increase cheating frequencies. For high fines joint liability will never backfire and can reduce the rate of cheating if the fine is appropriately shared. Figure \ref{fig: cheatingprob} visualizes the Theorem by dividing the parameter space into different regions according to the resulting equilibrium cheating probabilities.   

The intuition for our result can be nicely deduced from the figure. When contestants bear full liability for fines for detected cheating $\eta=1$ they do not cheat. Both managers will cheat, since they have nothing to fear even if the team is caught cheating. Therefore, we obtain an equilibrium cheating probability of $1/2$. Moving some liability from the contestant to the manager reduces the cheating incentive for the manager and increases cheating incentives for the contestants at the same time. If the fine is relatively low, then at least one of the contestants finds it profitable to cheat before the higher liability put on the managers prevents them from cheating. As a consequence, the aggregate cheating probability increases. In contrast, a high fine leads the manager to stop cheating before the contestant starts to do so, which reduces the probability of equilibrium cheating.

\section{\label{sec:Experimental-Design4}Experimental design, hypotheses, and procedures}
Our experiments are designed to test the general results of our model and to answer two fundamental questions: First, can joint liability effectively reduce cheating rates in environments with strong deterrence, as indicated by theoretical predictions? Second, does, as theory suggests, joint liability backfire and lead to more cheating if fines are lenient?

\subsection{Design and hypotheses}
In Figure \ref{fig: cheatingprob}, we mark the location of our four main treatments with solid circles. Our additional treatments which serve the purpose of robustness checks (discussed in Section \ref{robustness}) are represented by hollow circles. In the model Section, we normalized the prize to unity. For the experiment, we set a prize of $V=90$ for all treatments. This is without loss of generality, as it just scales contest revenue, equilibrium efforts and does not alter $\hat{\eta}$ and $\hat{\mu}$ as long as the expected fine is scaled such that the ratio to the to the total prize remains the same.\footnote{More specifically, in the analysis above the expected fine $pf$ is to be replaced by half the expected fine per dollar of the prize $pf/2V$.} Furthermore, we set the cheating efficiency $\delta$ equal to two and the fraction of revenue received by the contestant $r$ to $2/3$. The detection probability is chosen to be 25 percent.

In our initial $2\times2$ design, we vary the cheating incentives by using two different fine levels $f_L=55$ and $f_H=65$. This implies relative expected fines per unit of the prize of $11/72$ and $13/72$, respectively.
In the other dimension, we vary the liability share $\eta$. In the individual-liability regime $Ind$, the fine for caught cheating is borne by the contestant alone ($\eta=1$). Under joint liability $Jo$, the manager and the contestant share not only the revenue from the contest, but also the fine if cheating is detected. We set the proportion of the fine borne by the contestant equal to the fraction of revenue that she receives ($\eta=r=2/3$). This results in four treatments:
\textit{Jo\_H}, \textit{Jo\_L}, \textit{Ind\_H}, and \textit{Ind\_L}.

The parameters of the four treatments are chosen such that moving from individual liability to joint liability, in theory, should have a different impact on aggregate cheating rates for the two fine levels. For both fine levels under individual liability in equilibrium the managers always cheat, while the contestants do not. Therefore, the resulting average cheating probability is 50 percent under both individual-liability treatments. For a high fine, joint liability should reduce the cheating probability to 25 percent, as the managers do not cheat anymore, while on average only one of the two contestants cheats.\footnote{There are multiple equilibria that yield the same average cheating probability but differ in the equilibrium strategies of the contestants. Either one of the two contestants cheats with certainty, or both randomize with equal probability.} Under the low fine, theory predicts joint liability to backfire. Managers are still cheating, even after they are liable for a share of the fine, while on average in equilibrium one of the two contestants cheats. The resulting equilibrium cheating probability increases from 50 to 75 percent. We summarise the theoretical predictions in Table \ref{tab:Asummaryofparameters}.

\begin{table}[H]
\centering

\begin{tabular}{lcccc}
\hline\hline 
Predicted cheating rate & \textit{Jo\_H} & \textit{Jo\_L} & \textit{Ind\_H} & \textit{Ind\_L} \\
\hline 
Avg. contestant cheating rate & 50\% & 50\% & 0\% & 0\% \\
Manager cheating rate & 0\% & 100\% & 100\% & 100\% \\
Avg. overall cheating rate $\bar{p}$ & 25\% & 75\% & 50\% & 50\% \\
\hline\hline
\end{tabular}

\caption{Predicted Cheating Behaviour across Treatments\label{tab:Asummaryofparameters}}
\end{table}

From the theory we derive the following directional hypotheses. 

\begin{hypothesis}
\label{concheateqm} The average cheating probabilities satisfy the following orderings: 
\begin{enumerate}
\item No impact of the fine with individual liability: $\bar{p}$(\textit{Ind\_H)}$=\bar{p}$(\textit{Ind\_L)}
\item Joint liability reduces cheating for high fines: $\bar{p}$(\textit{Jo\_H)}$<\bar{p}$(\textit{Ind\_H)}
\item Joint liability backfires for low fines: $\bar{p}$(\textit{Jo\_L)}$>\bar{p}$(\textit{Ind\_L)}.
\end{enumerate}
\end{hypothesis}

Moving from the aggregate cheating probability to individual behavior by role and liability scheme, we put forward two role-specific hypotheses arising from the equilibrium predictions.

\begin{hypothesis}
\label{conind}Under the individual liability scheme (\textit{Ind\_L} and
\textit{Ind\_H}), contestants never cheat while managers always do. 
\end{hypothesis}

\begin{hypothesis}
\label{conjo}Under joint liability, managers cheat
in the treatment \textit{Jo\_L} but not in the \textit{Jo\_H} treatment, while the size of the fine does not affect the cheating behavior of the contestants.
\end{hypothesis}

Finally, we turn to efforts. In general, designers of contest environments are interested in high efforts.
From Equation (\ref{eq:estar}) we know that efforts are higher in equilibrium if the contests are even because no team has the edge by cheating unilaterally. In the pure-strategy equilibria of the game, the probability of matching cheating decisions is identical at one-half across all treatments.\footnote{Recall that the team cheating probabilities are $50\%$ for both teams in the individual liability treatments, which yields probabilities of $25\%$ for both teams cheating as well as for both teams not cheating. In \textit{Jo\_L}, one team cheats with certainty, while the other cheats with probability $50\%$, which yields a probability of both teams cheating of $50\%$, while both teams never end up not cheating. Similarly, in \textit{Jo\_H} the probability of both teams not cheating is $50\%$, while it never happens that both cheat.} 
However, for the joint-liability treatments there exists also a mixed-strategy equilibrium, where both teams have identical cheating probabilities ($75\%$ in \textit{Jo\_L} and $25\%$ in \textit{Jo\_H}), which both lead to a higher probability of matching cheating decisions at $5/8$. Hence, depending on which equilibrium is selected, efforts are expected to be either identical or greater in the joint liability treatments.

\begin{hypothesis}
\label{coneffort} Efforts are weakly greater in the joint-liability treatments.
\end{hypothesis}

\subsection{Procedures}
The experiment took place at the Adelaide Laboratory for
Experimental Economics (AdLab) at the University of Adelaide. For
each treatment, we ran three sessions. All 12 sessions were programmed
and implemented using Z-tree (\citealp{fischbacher2007z}). 248 subjects
were recruited using the online recruitment system ORSEE (\citealp{greiner2015subject}).
Each subject participated in one session only. 

Upon arrival, the participants were randomly assigned to a computer. The instructions were distributed and read aloud by the experimenter. To ensure the understanding of the decision tasks, subjects were required to correctly answer control questions before starting the actual experiment. In each session, subjects were placed into teams of two. We randomly assigned roles (one contestant and one manager per team), which remained the same throughout the experiment. Two teams were randomly matched to compete against each other for 20 consecutive rounds. In the first stage, both team members chose their individual cheating intentions. One of the intentions was randomly selected and implemented as the team's final cheating decision for the round.
Subsequently, the contestants and managers were informed of the four individual cheating intentions and the two implemented decisions. In the second stage, the contestants decided their effort level. At the end of each round, all subjects received feedback if their team was caught cheating and were informed of the efforts and resulting prize shares for both teams. Furthermore, the two team members were informed about each other's payoffs. All earnings were expressed in Experimental Currency Units, which we converted into Australian dollars at the end of the experiments at an exchange rate of one dollar per every ten ECUs. A session lasted about one hour and 20 minutes, and participants earned 26.23 AUD on average. Earnings included a show-up fee of 5 dollars. 

\section{\label{sec:Results4}Results}
 In this section, we present the main results. We start with descriptive summary
statistics on the prevalence of cheating. After that, we investigate
the aggregate dynamics of cheating and then use regression analysis
to gain insight into the drivers behind the aggregates and to conduct statistical tests. Finally, we take a closer look at the impact of treatments on effort exertion and competitiveness.

\subsection{Aggregate cheating rates }
Table \ref{tab:Average-and-standard} reports the average proportion
of cheating decisions by treatment and role. We find that in the three treatments \textit{Ind\_H}, \textit{Ind\_L}, and \textit{Jo\_L} slightly more than $60\%$ of all decisions involve cheating. Notably, in the \textit{Jo\_H} treatment, this rate is lower at $53\%$. Hence, for high fines, joint liability (\textit{Jo\_H}) seems to outperform individual liability (\textit{Ind\_H}). However, the difference between the two regimes becomes negligible when fines are low, indicating that joint liability does not backfire with low fines as predicted by theory. 

Next, we break down the overall cheating rate into cheating decisions made by contestants and managers. Among contestants, we observe higher cheating rates in the \textit{Jo\_L} ($62.3\%$) and \textit{Jo\_H} ($50\%$) treatments compared to the \textit{Ind\_L} ($44.6\%$) and \textit{Ind\_H} ($41.7\%$) treatments. This observation aligns with theoretical predictions, indicating that joint liability creates greater cheating incentives for contestants. Managers respond to joint liability with decreased cheating. The cheating rate among managers in the \textit{Jo\_L} (62\%) and \textit{Jo\_H} (56\%) treatments is lower compared to over 80\% in the \textit{Ind\_L} and \textit{Ind\_H} treatments. The observation that the reduction occurs for both fine levels is contradicting the theory that predicts only a reduction when fines are high.

For a first summary, we find that the observed behavior does not conform to crisp equilibrium play. However, the data exhibit consistent patterns that partially agree with the comparative statics of the
theory. Contestants exhibit the expected directional changes in their cheating likelihood between different punishment regimes. The managers' behavior is less consistent with theory. The theory suggests that for subgame-perfect efforts, managers have dominant strategies to cheat in the \textit{Jo\_L} treatment and not to cheat in \textit{Jo\_H}. Yet, in the experiment, the observed difference in managers' cheating rates between these two treatments is small, differing by only 0.06 percentage points. Managers tend to reduce cheating irrespective of the fine size when they have to bear liability. 

\begin{table}[H]
\centering
\begin{tabular}{ m{2cm} >{\centering\arraybackslash}m{1.5cm} >{\centering\arraybackslash}m{2cm} >{\centering\arraybackslash}m{2cm} >{\centering\arraybackslash}m{2cm} }
\hline
\hline
Treatment & \textit{N} & \begin{tabular}{@{}c@{}}Overall\\Cheating\\Rate\end{tabular} & \begin{tabular}{@{}c@{}}Contestant\\Cheating\\Rate\end{tabular} & \begin{tabular}{@{}c@{}}Manager\\Cheating\\Rate\end{tabular} \\
\hline
\textit{Jo\_H} & 1200 & \begin{tabular}{@{}c@{}}0.530\\(0.499)\end{tabular} & \begin{tabular}{@{}c@{}}0.500\\(0.500)\end{tabular} & \begin{tabular}{@{}c@{}}0.560\\(0.497)\end{tabular} \\
\textit{Jo\_L} & 1120 & \begin{tabular}{@{}c@{}}0.621\\(0.485)\end{tabular} & \begin{tabular}{@{}c@{}}0.623\\(0.485)\end{tabular} & \begin{tabular}{@{}c@{}}0.620\\(0.486)\end{tabular} \\
\textit{Ind\_H} & 1280 & \begin{tabular}{@{}c@{}}0.610\\(0.488)\end{tabular} & \begin{tabular}{@{}c@{}}0.417\\(0.493)\end{tabular} & \begin{tabular}{@{}c@{}}0.803\\(0.398)\end{tabular} \\
\textit{Ind\_L} & 1360 & \begin{tabular}{@{}c@{}}0.624\\(0.485)\end{tabular} & \begin{tabular}{@{}c@{}}0.446\\(0.497)\end{tabular} & \begin{tabular}{@{}c@{}}0.801\\(0.399)\end{tabular} \\
\hline
\end{tabular}
\caption{Average cheating rates by Treatment with standard deviations in parentheses.  \label{tab:Average-and-standard}}
\end{table}

\subsection{Cheating dynamics} \label{cheating dynamics}
The analysis above, is based on averages of groups of four over 20 periods. Next we investigate how behavior develops over time. In the strategically rich and complex environment, we expect participants to change their behavioor over time as they learn and adapt their behavior to their teammate's and the opposition team's behavior. To explore how cheating behavior evolves over time, we analyze the average cheating rate across four treatments, as depicted in Figure \ref{fig:Proportion-of-cheating}.

The left panel of the figure displays the average cheating rate per period. As expected, there is a large variation between periods. The \textit{Jo\_L} treatment consistently records the highest cheating rates whereas the \textit{Jo\_H} treatment maintains the lowest cheating rates. 


In the right panel, we eliminate the period-to-period variation by aggregating five periods. This reveals that the \textit{Jo\_L}, \textit{Ind\_L} and \textit{Ind\_H} treatments yield similar overall cheating rates over all four quarters, while the \textit{Jo\_H} treatment exhibits a lower cheating rate for all four quarters. Furthermore, we observe a downward trend of the average cheating probability in the \textit{Jo\_H} treatment, whereas other three treatments do not exhibit a time trend.

\begin{figure}[H]
\begin{centering}
\includegraphics[scale=0.35]{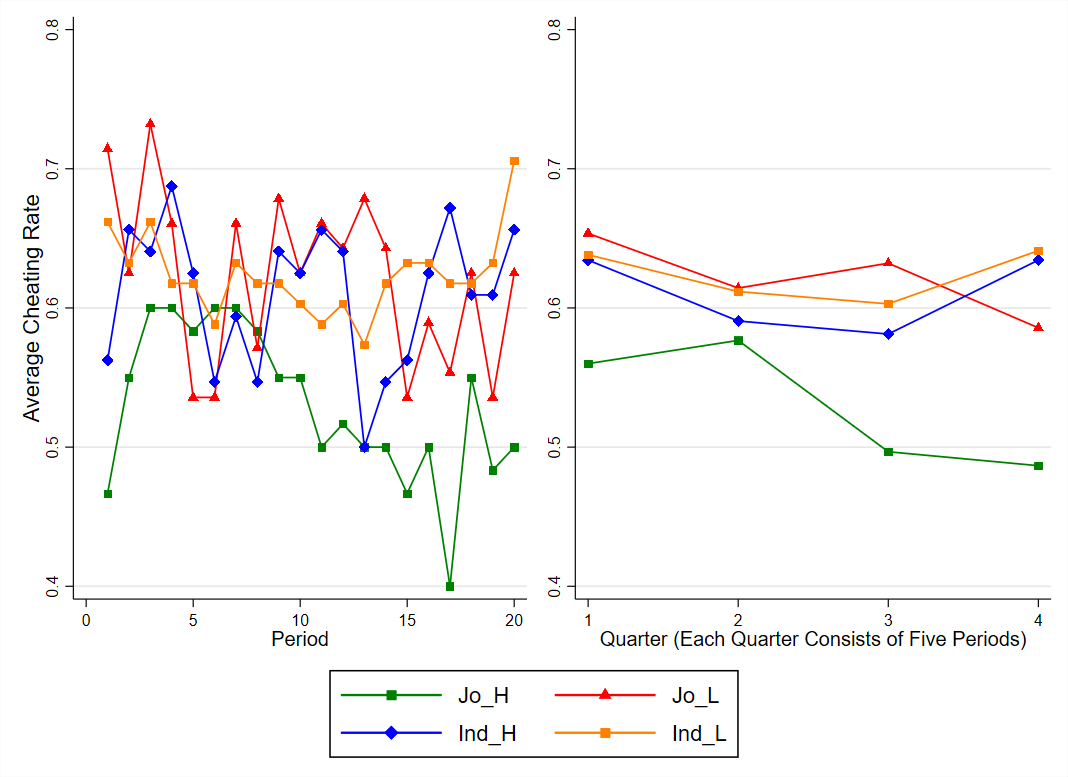}
\par\end{centering}
\caption{Average Cheating Rate over Time by Treatment\label{fig:Proportion-of-cheating}}
\end{figure}

\subsection{Cheating patterns within teams} 

In each period, there are four combinations of cheating decisions a team can make. Either no one in a team cheats, both cheat, or only the manager or the contestant cheats. A plot of relative frequencies of the four combinations (gray bars in Figure \ref{fig:Histogram-of-types}) provides information on the treatment-specific intra-team cheating dynamics. The equilibrium fractions are depicted as transparent boxes without fill.

\begin{figure}[h!]
\begin{centering}
\includegraphics[scale=0.25]{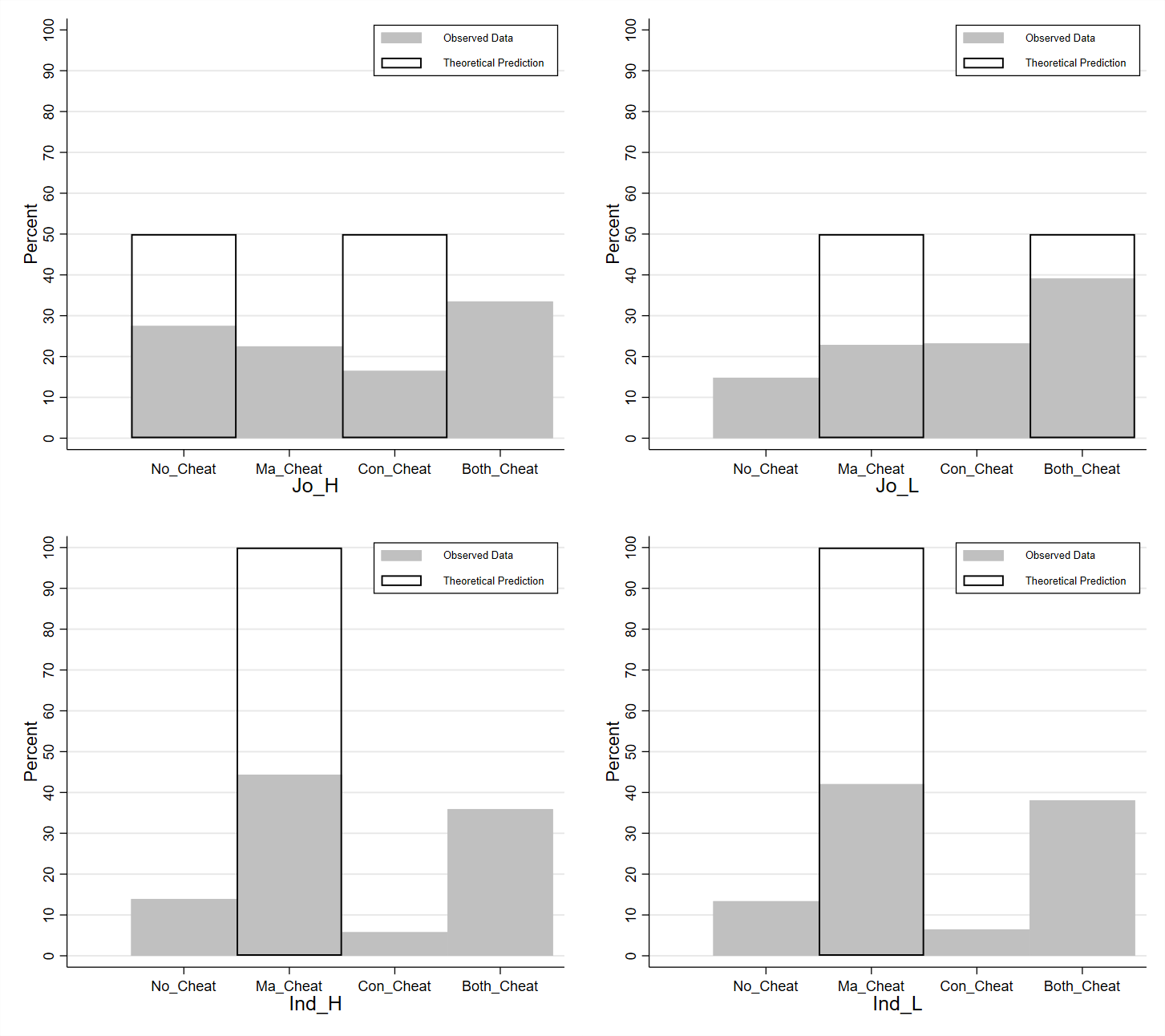}
\par\end{centering}
\caption{Distribution of Team Types by Treatment \label{fig:Histogram-of-types}}
\end{figure}

Play in all four treatments differs markedly from equilibrium behavior. Moreover, there is a strong over-all impact of the treatments on behaviour ($p=0.008$ $\chi^2$ test). In the individual liability treatments (bottom row) we often see both players instead of, as predicted, only the manager cheating. Managers have a dominant strategy to cheat, but in both treatments about 20 percent don't play the dominant strategy. In the joint liability treatments (top row), we see that the combinations which we should never observe (only the manager or both cheat with a high fine; nobody or only the contestant cheats with a low fine) consistently occur. 

Despite the clear deviations from theoretical predictions, we see that incentives matter, since the comparative statics with respect to changes in the fine rate for a given liability scheme are consistent with theory. Under individual liability, theory predicts that a change in the fine rate does not change cheating behavior. This is visually confirmed. The distributions do not appear to be different. Under joint liability, moving from a low fine (top left panel) to a low fine (top right panel), we see that the distributions differ considerably. In particular, we see a shift from the combinations in which no one cheats to those in which both players cheat. 

Comparing the cheating profiles under joint liability with those under individual liability reveals clear differences. Under both fine regimes, as expected, joint liability reduces the instances of only the manager cheating, while the situation where only the contestant cheats arises more often.

\subsection{Treatment effect on cheating}

In what follows, we formalize the analysis by using regressions. We employ a random-effects multinomial logit with a team's cheating choice as the dependent variable. Treatment dummies are added as independent variables to estimate treatment effects\footnote{Full tables with estimated coefficients are provided in the appendix.}. To address the potential correlation of cheating decisions within groups (i.e., the two teams matched for the whole experiment) over time, we allow for clustering of the error term at the group level. As discussed in subsection \ref{cheating dynamics}, early-period choices tend to be noisy. Therefore, we run the regression model separately for the first ten and the second ten periods. In the second half of the experiments, participants are  more familiar with the game and the strategies of the opposing team, which leads to behavior settling down and noise being reduced.
\begin{figure}[H]
\begin{centering}
\includegraphics[scale=0.35]{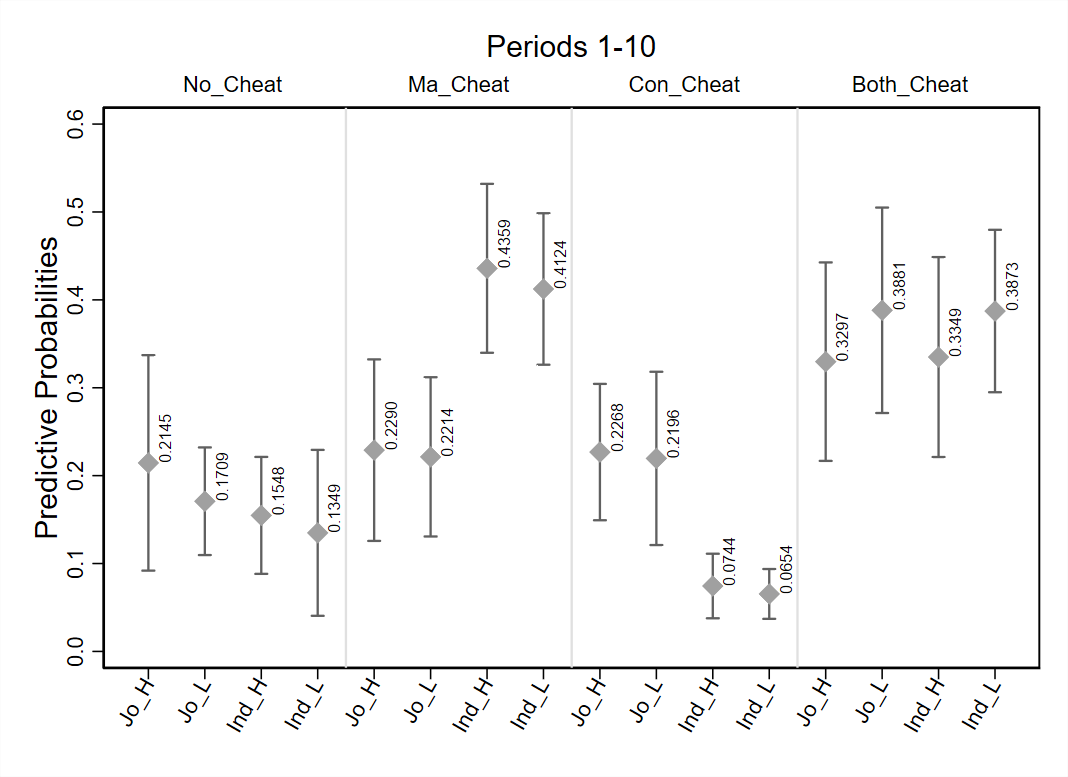}
\par\end{centering}
\caption{Predictive Probabilities of Team Type by Treatment (First Half) \label{fig:xtmlogitfirsthalf}}
\end{figure}

Figure \ref{fig:xtmlogitfirsthalf} plots the predicted probabilities of observing the different cheating profiles within a team by treatment for the first 10 periods. The graph comprises four panels, each representing a specific cheating profile with point-estimates and confidence bands. In the first ten periods there are no significant treatment differences for the probabilities of no one or both team members cheating. As expected, the probability of teams where only the manager cheats is significantly higher under individual liability than under joint liability ($p<0.01$ for all pairs of the \textit{Jo} and \textit{Ind} treatments).

The greater likelihood of just the manager cheating under individual liability comes at the expense of the contestant cheating alone The predicted probability that only the contestant cheats is significantly higher under joint liability ($p<0.01$ for all combinations between the \textit{Jo} and \textit{Ind} treatments). Consequently, in the first 10 periods, the introduction of joint liability has the main effect that the likelihood of observing only the contestant cheating increases by the same magnitude as the likelihood that only the manager cheats decreases, which leaves the overall cheating probability unchanged. Hence, regardless of the size of the fine, we do not find a treatment effect of introducing joint liability on team cheating rates in early rounds.

\begin{figure}[H]
\begin{centering}
\includegraphics[scale=0.35]{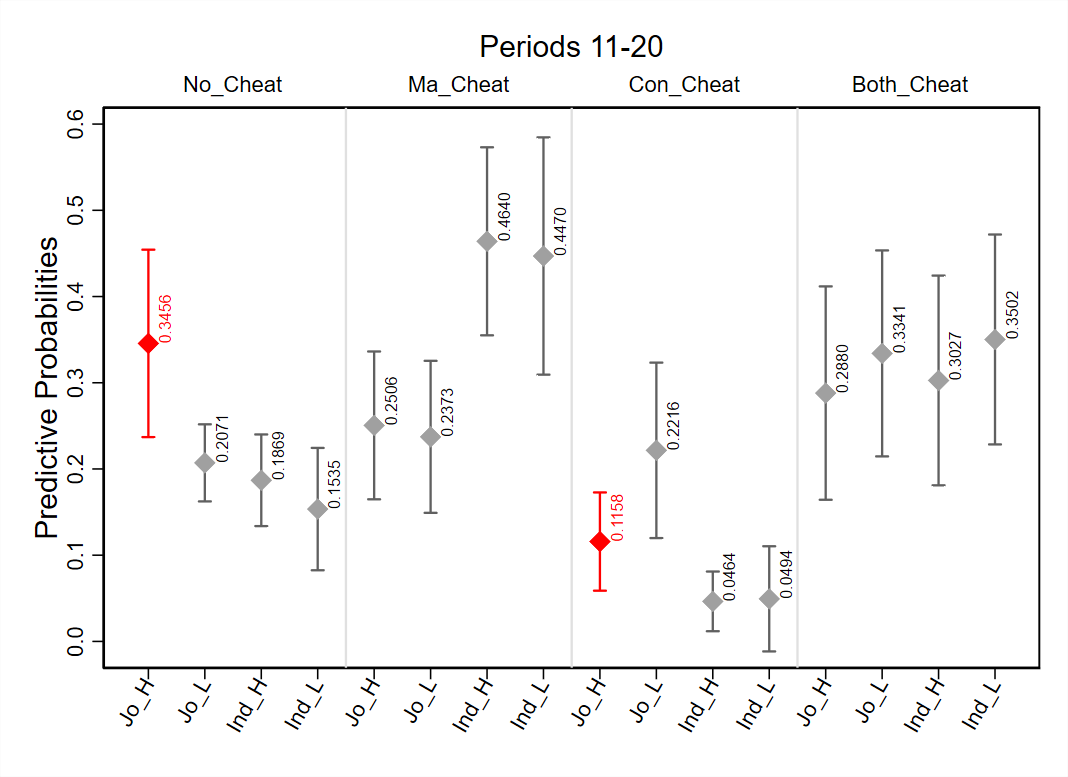}
\par\end{centering}
\caption{Predicted Probabilities of Team Type by Treatment (Second Half) \label{fig:xtmlogitsecondhalf}}
\end{figure}

Figure \ref{fig:xtmlogitsecondhalf} shows that once behavior has settled (in periods 11 to 20), a treatment effect appears.  Surprisingly, the \textit{Jo\L} treatment is the only treatment where behavior changes substantially over time. Comparing the results between earlier and later rounds in this treatment reveals one behavioral shift that results in joint liability reducing cheating if fines are high. In the \textit{Jo\_H} treatment, we observe a higher probability that both team members refrain from cheating. Teams under the \textit{Jo\_H} treatment exhibit a $34.67\%$ predicted probability of landing in this  category, which is higher than in the initial rounds and also significantly higher than in the other three treatments in the later rounds ($p = 0.009$ for \textit{Ind\_H} vs. \textit{Jo\_H}, and $p = 0.003$ for \textit{Ind\_L} vs. \textit{Jo\_H} and $p = 0.018$ for \textit{Jo\_L} vs. \textit{Jo\_H}). The increased likelihood of no one cheating under joint liability with a high fine in later rounds comes at the expense of only the contestant cheating (declining from $22.68\%$ in the first half to $11.58\%$ in the second half). Hence, the shift that led to a treatment effect originates from fewer contestants cheating when the manager does not cheat either.

Since cheating rates were already similar in the low-fine treatments in the early rounds and there were no changes over time, we do not find a treatment effect on over-all cheating between the two liability regimes when fines are low. The theoretically predicted backfiring of joint liability under low fines did not occur. Average predicted cheating probabilities are ordered by treatment as follows: $\bar{p}$(\textit{Jo\_L})$=\bar{p}$(\textit{Ind\_L})$=\bar{p}$(\textit{Ind\_H})$>\bar{p}$(\textit{Jo\_H}).

Using the regression results we can estimate a composite cheating probability for each treatment and test differences for significance.\footnote{The expected cheating probability is the predicted probability that both cheat plus half the predicted probability that only the manager cheats plus half the probability that only the contestant cheats.} In the second half of the experiments the estimated difference between the cheating probabilities in the joint and individual treatments with high fines is -0.1092, which is marginally significantly smaller than zero ($p=0.073$). In contrast there is no significant difference between the predicted cheating probabilities in the treatments with low fines (-0.0027, $p = 0.964$). 

\newtheorem{remark}{Result}

\begin{remark}
\emph{Joint liability outperforms individual liability in deterring cheating behavior if the fine is high, as it encourages both team members to refrain from cheating. However, with a low fine, joint liability does not cause different cheating rates compared to individual liability. 
\label{fig:mlogitfirsthalf}}
\end{remark}

\begin{remark}
\emph{The observed superior performance of the \textit{Jo\_H} treatment in deterring overall cheating can be attributed to contestants reducing cheating, particularly in the later periods of the experiment.} 
\end{remark}


Considering teams involving cheating managers (i.e. both or only the manager cheats), we do not find any differences across the joint liability treatments or across the individual-liability treatments. This confirms our earlier observation that managers do not adjust their cheating behavior in response to changes in fine severity. This is expected for the individual liability treatments but not for the joint liability treatments. While the fine should not impact the managers cheating decision the liability regime should. The predicted probability of only the manager cheating is consistently and significantly lower under joint liability than under individual liability ($p=0.002$ for \textit{Jo\_H} versus \textit{Ind\_H},  and $p=0.011$ for \textit{Jo\_L} versus \textit{Ind\_L} in the later periods). 

\begin{remark}
\emph{Managers cheat less under joint liability regardless of the severity of the fine.} 
\end{remark}

\subsection{Aggregate effort and its dynamics}

Moving to the second stage of the game, we analyze the effort decisions by contestants.
Table \ref{tab:Average-and-standard-efforts} reports the average
efforts and the degree to which these efforts deviate from the theoretically optimal level. To calculate the optimal effort, we use the realized cheating decisions of both teams (13.33 if only one team cheats and 15 if neither or both teams cheat). The deviation from optimality is then calculated by taking the difference between the actual effort and the subgame perfect continuation effort. A positive deviation implies over-exertion of effort, while a negative value implies under-exertion. 

 The average efforts are similar in the four treatments with a minimum of 16.50 in the \textit{Jo\_L} treatment and a maximum of 18.87 in the \textit{Ind\_H} treatment. These averages are higher than the optimal levels. Additionally, in all treatments, the mean values of the deviation from optimality are positive, confirming a consistent trend of over-exertion of efforts in the contests. This finding aligns with the existing contest literature \citealp{chaudhuri2011sustaining,choi2007coevolution,sheremeta2018behavior}.

\begin{table}[H]
\centering
\begin{tabular}{ m{2cm} >{\centering\arraybackslash}m{1.5cm} >{\centering\arraybackslash}m{2cm} >{\centering\arraybackslash}m{3cm} }
\hline
\hline
Treatment & \textit{N} & \begin{tabular}{@{}c@{}}Average\\Effort\end{tabular} & \begin{tabular}{@{}c@{}}Average Deviation\\from Optimally\end{tabular} \\
\hline
\textit{Jo\_H} & 600 & \begin{tabular}{@{}c@{}}18.8700\\(15.5164)\end{tabular} & \begin{tabular}{@{}c@{}}4.5367\\(15.5178)\end{tabular} \\
\textit{Jo\_L} & 560 & \begin{tabular}{@{}c@{}}16.5000\\(7.9531)\end{tabular} & \begin{tabular}{@{}c@{}}2.1964\\(7.9346)\end{tabular} \\
\textit{Ind\_H} & 640 & \begin{tabular}{@{}c@{}}18.8797\\(13.9484)\end{tabular} & \begin{tabular}{@{}c@{}}4.6089\\(13.8229)\end{tabular} \\
\textit{Ind\_L} & 680 & \begin{tabular}{@{}c@{}}17.7941\\(13.1882)\end{tabular} & \begin{tabular}{@{}c@{}}3.5049\\(13.1840)\end{tabular} \\
\hline
\end{tabular}
\caption{Average Effort and Over-dissipation
\label{tab:Average-and-standard-efforts}}
\end{table}

Figure \ref{fig:Average-effort-by} depicts the average effort by periods (left panel) and the average by quarters (right panel). In the first period, the average efforts start at about 22 to 27 and decrease sharply in the first quarter. The extremely high efforts in early periods, is likely do to confusion among some subjects in the experiment's early stages. Once they receive payoff feedback and realize that over-exertion leads to low or even negative profits, they swiftly adjust their efforts. From the second quarter onward, the effort trends upwards again in the \textit{Jo\_H} and \textit{Ind\_H} treatments, while it declines further in the \textit{Jo\_L} treatment.   

\begin{figure}[H]
\begin{centering}
\includegraphics[scale=0.35]{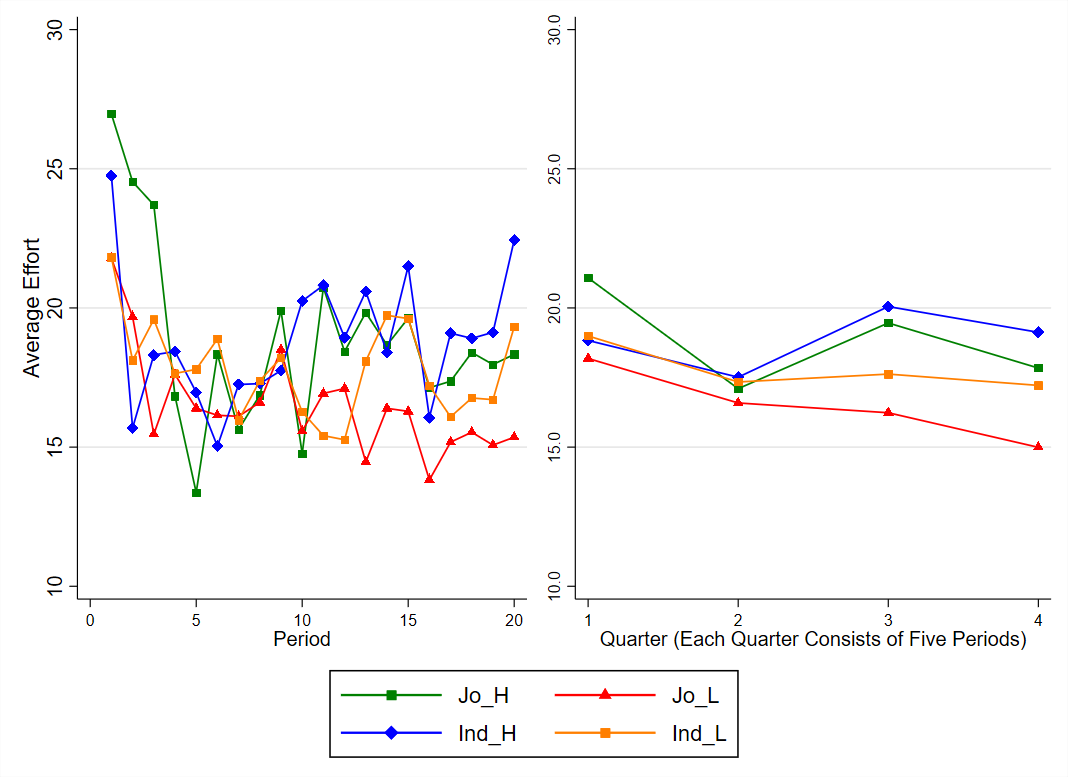}
\par\end{centering}
\caption{Average Effort over Time by Treatment \label{fig:Average-effort-by}}
\end{figure}

\subsection{Treatment effect on over-dissipation}

It is important to note that achieving the equilibrium effort requires both contestants to best respond to their rival's effort. However, if one contestant persistently deviates from the best response level, perhaps due to reasons such as confusion or a desire to win regardless of costs, then this disrupts the equilibrium. Consequently, the other contestant's best response is affected, impeding convergence towards equilibrium. Over the course of 20 periods of repeated contests with fixed pairs of contestants, we examine whether contestants adjust their efforts based on their competitors' effort levels. Figure \ref{fig:Effort-versus-best} illustrates contestants' actual effort plotted against their best response to their competitor's effort. The determination of best response effort involves substituting the opponent's actual effort and the implemented cheating decisions of both teams into the derived best response equation \ref{eq:breff}. A quadratic fitting line illustrates the observed relationship between these two variables.

\begin{figure}[ht]
\centering
\includegraphics[scale=0.23]{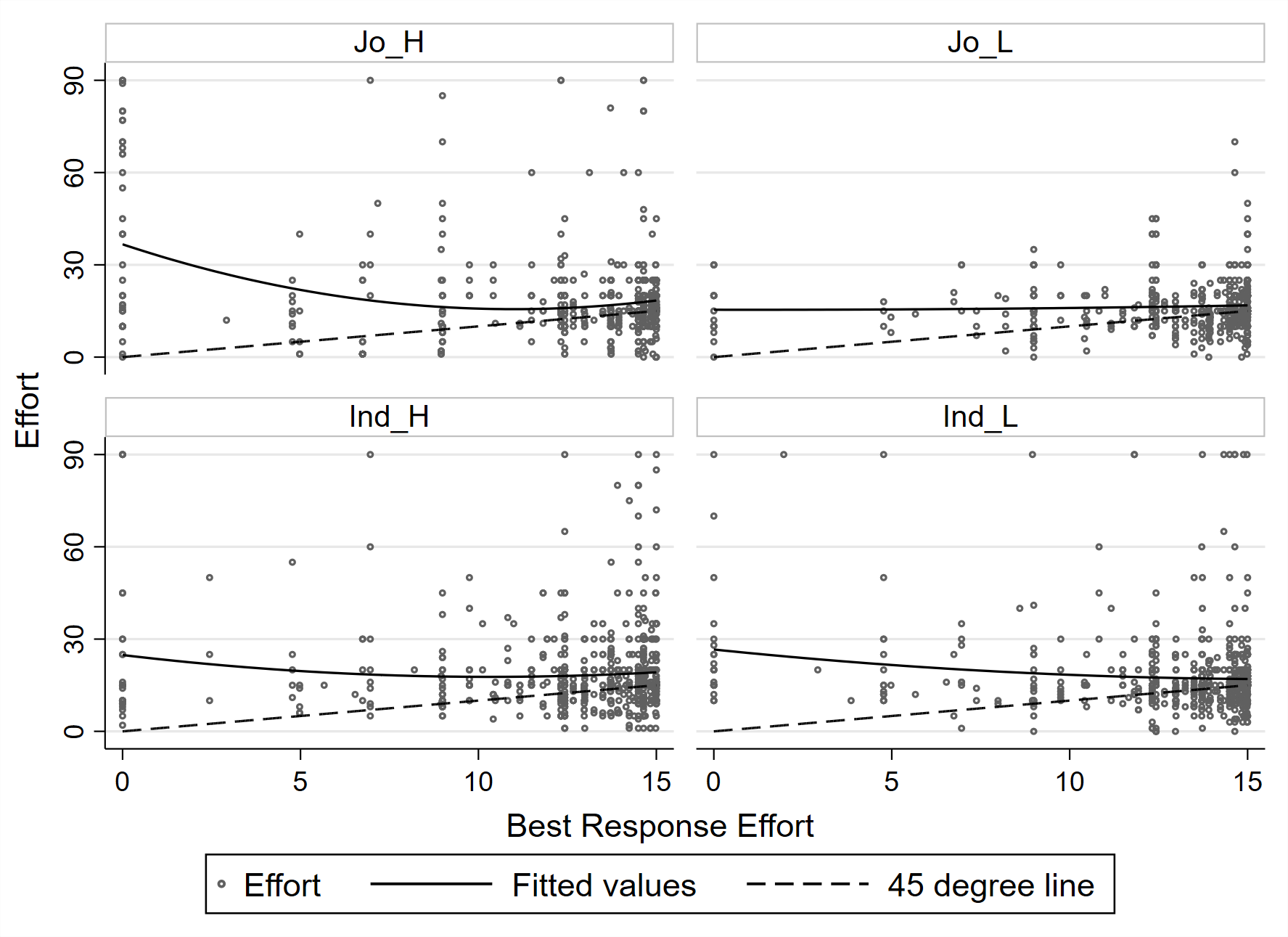}
\caption{Observed Effort versus Best Response Effort by Treatment \label{fig:Effort-versus-best}}
\end{figure}

In all treatments, the quadratic prediction consistently lies above the 45-degree line, indicating persistent over-exertion. As the best response effort approaches zero, the deviation between one's actual effort and the best response level widens. This suggests that when an opponent displays irrational over-competitiveness by exerting exceptionally high effort, the contestant from the other team tends to respond with an effort that exceeds the best response effort by far. This contradicts the prediction, where an extremely high effort by one party should prompt minimal effort by the other. We observe repeated instances of these hyper-competitive periods between rivals which can be interpreted as ongoing feuds.

The severity of over-exertion lessens when the effort of the other team's contestant lies in the equilibrium range of 13 to 15. This pattern is depicted by the reduced gap between the 45-degree line and the quadratic fitting line. The majority of observations lie in this area.

Next we search for factors that contribute to the over-dissipation. Given the interdependence between the contestants' effort decisions, we conduct a group-level analysis and derive the average over-dissipation for each contest. The average over-dissipation is calculated as the difference between the average effort exerted by paired contestants and the equilibrium effort. We then conduct a random-effects regression on the average over-dissipation. The regression model incorporates treatment dummies, the total count of cheating teams within each contest, and quarterly time dummies, with clustering on the group level to account for intra-group correlations over time. In addition to analyzing the first and second halves of the experiment separately, we include a model for the full experiment and add quarter dummies.

The estimated coefficients of these three models are presented in Table \ref{tab:random-effect-over-dissipation}. In the full model, we find an obvious time trend, indicating the strongest over-dissipation takes place in the initial quarter, with significantly reduced levels observed in the subsequent quarters ($p=0.009$ for the second quarter, and $p=0.090$ for the last quarter.) Moreover, the level of over-dissipation  is significantly different from zero for three treatments at the 5\% (\textit{Jo\_L}, \textit{Ind\_H}, and \textit{Ind\_L}, with $p$-values of 0.041, 0.009, and 0.013) and significant at the 10-percent level for the \textit{Jo\_H} treatment ($p$-value = 0.066). However, we do not observe significant differences in over-dissipation across treatments ($p>0.1$ for all treatment pairs). 

 We also find that over-exertion is largest when both teams cheat ($p=0.037$ for both team cheating versus no team cheating groups; $p=0.077$ for both team cheating versus one team cheating groups). The observed effort choices in the split regressions is roughly consistent with the regression for the full experiment. The strong effect of both teams cheating on over-exertion remains.

\begin{remark}
\emph{Over-exertion of efforts is strongest if both teams cheat and does not significantly differ between treatments} 
\end{remark}

\begin{table}[H]
\centering
\begin{tabular}{ m{4.5cm} >{\centering\arraybackslash}m{2.5cm} >{\centering\arraybackslash}m{2.5cm} >{\centering\arraybackslash}m{2.5cm} }
\hline
\hline
Over-dissipation & \begin{tabular}{@{}c@{}}Full\\Model\end{tabular} & \begin{tabular}{@{}c@{}}First-Half\\Model\end{tabular} & \begin{tabular}{@{}c@{}}Second-Half\\Model\end{tabular}\\
\hline
Treatment (Base:\textit{Jo\_H}) & & & \\ 
\textit{Jo\_L} & \begin{tabular}{@{}c@{}}-2.5636\\(2.7477)\end{tabular} & \begin{tabular}{@{}c@{}}-1.7376\\(2.4729)\end{tabular} & \begin{tabular}{@{}c@{}}-3.2986\\(3.2650)\end{tabular} \\
\textit{Ind\_H} & \begin{tabular}{@{}c@{}}-0.0957\\(3.0868)\end{tabular} & \begin{tabular}{@{}c@{}}-0.8697\\(2.7297)\end{tabular} & \begin{tabular}{@{}c@{}}0.7104\\(3.8386)\end{tabular} \\
\textit{Ind\_L} & \begin{tabular}{@{}c@{}}-1.2085\\(2.8927)\end{tabular} & \begin{tabular}{@{}c@{}}-0.9665\\(2.7588)\end{tabular} & \begin{tabular}{@{}c@{}}-1.4156\\(3.6275)\end{tabular} \\
Number of Cheating Teams (Base: No Team Cheats) & & & \\ 
\textit{One Team Cheats} & \begin{tabular}{@{}c@{}}1.0025\\(1.0480)\end{tabular} & \begin{tabular}{@{}c@{}}0.4817\\(1.3898)\end{tabular} & \begin{tabular}{@{}c@{}}1.3010\\(0.9532)\end{tabular} \\
\textit{Both Teams Cheat} & \begin{tabular}{@{}c@{}}2.1571**\\(1.0337)\end{tabular} & \begin{tabular}{@{}c@{}}1.4220\\(1.5725)\end{tabular} & \begin{tabular}{@{}c@{}}1.9316**\\(0.9160)\end{tabular} \\
Quarters (Base: Q1; Q3) & & & \\ 
\textit{Q2} & \begin{tabular}{@{}c@{}}-1.9786***\\(0.7623)\end{tabular} & \begin{tabular}{@{}c@{}}-2.0260***\\(0.7855)\end{tabular} & \\
\textit{Q3} & \begin{tabular}{@{}c@{}}-0.8625\\(1.1761)\end{tabular} & & \\
\textit{Q4} & \begin{tabular}{@{}c@{}}-1.8931*\\(1.1171)\end{tabular} & & \begin{tabular}{@{}c@{}}-1.0342*\\(0.5476)\end{tabular} \\
Intercept & \begin{tabular}{@{}c@{}}4.6002*\\(2.7319)\end{tabular} & \begin{tabular}{@{}c@{}}5.0872*\\(2.8423)\end{tabular} &  \begin{tabular}{@{}c@{}}3.7241\\(3.0030)\end{tabular} \\
\hline
\multicolumn{4}{l}{{\small{}{*}{*}{*} Sig. at the 1 percent level, {*}{*}Sig. at the
5 percent level, {*}Sig. at the 10 percent level}}\\
\end{tabular}
\caption{Random-effects Regression on Over-dissipation, Clustering on the Group Level
\label{tab:random-effect-over-dissipation}}
\end{table}

\section{Is the Absence of Backfiring Robust}\label{robustness}
Our experimental findings challenge the theoretically expected counterproductive effect of joint liability on cheating behavior when fines are lenient. This discrepancy could have arisen as the low fine of $55$ ECUs was not salient enough to entice the contestants to cheat once some of their liability was removed. As a robustness check, we ran two further treatments with a fine of $40$ ECUs \footnote{The decision to conduct additional treatments with the reduced fine was not part of the initial research plan and was in response to the observed results.}, referred to as \textit{Jo\_40} and \textit{Ind\_40}. 

This fine level was chosen to give joint liability the best chance to backfire. With a fine level of $40$ under joint liability all four players should choose to cheat in equilibrium, while under individual liability only the managers should cheat. Hence, the predicted equilibrium cheating rate should increase from $50\%$ to $100\%$ when joint liability is introduced.  See Figure \ref{fig: cheatingprob} for how the robustness treatments relate to the original treatments in the parameter space.

Our experimental data show a small aggregate difference with average cheating rates of $58.24\%$ and $63.24\%$ for the \textit{Ind\_40} and \textit{Jo\_40} treatments. Figure \ref{fig:robustcheatingtime} shows that this difference primarily arises from the earlier stages, particularly the first two quarters. In the \textit{Ind\_40} treatment, the cheating rate gradually increases over time. In contrast, the \textit{Jo\_40} treatment begins at a higher rate, decreases in the middle stages, and rises again by the final quarter. While the \textit{Jo\_40} treatment generally exhibits higher cheating rates, the difference between the two treatments diminishes over time, stabilizing at $1.76\%$ and $3.82\%$ in the final two quarters.
 
\begin{figure}[H]
\centering
\includegraphics[scale=0.3]{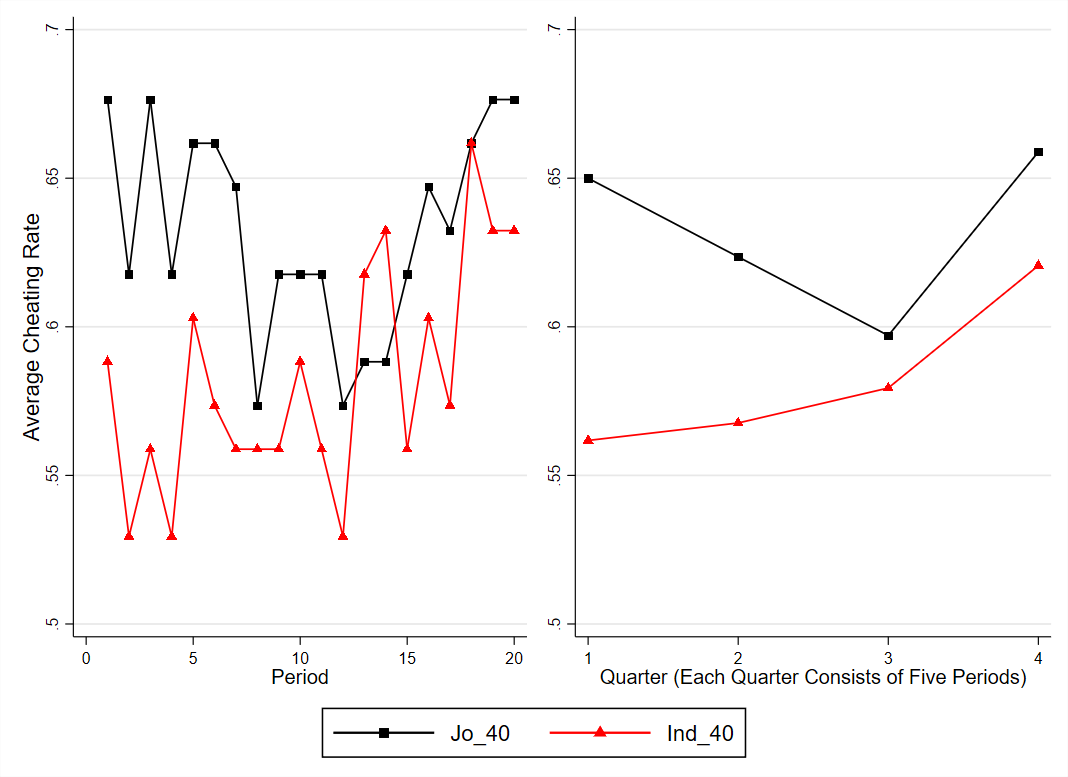}
\caption{Average Cheating Rate over Time by Additional Treatment
\label{fig:robustcheatingtime}}
\end{figure}

 Figure \ref{fig:robustcheatdistribution} illustrates team cheating decision distributions across the \textit{Jo\_40} and \textit{Ind\_40} treatments. We find a persistent influence of joint liability on cheating across different roles. Joint liability continues to discourage managers from engaging in cheating while simultaneously incentivizing increased contestant cheating. $74.41\%$ of managers cheat in the \textit{Jo\_40} treatment, while a higher proportion of $85.88\%$ engage in the \textit{Ind\_40} treatment. Conversely, $51.17\%$ of contestants cheat in the \textit{Jo\_40} treatment, while a lower proportion of $34.11\%$ partake in the \textit{Ind\_40} treatment. These findings challenge theoretical expectations once more, particularly in relation to the responsiveness of managers to the reduced fine of $40$.


Next we return to the frequency of observed combinations of intentions within teams. In the \textit{Jo\_40} treatment, both team members cheating is the modal combination of observed choices ($40\%$), while theory predicts it to be the only observed category. Similarly, in the \textit{Ind\_40} treatment, only the manager cheating is the most frequently observed ($52\%$) combination, without getting near to the theoretical rate of $100\%$. In summary, we can see a slight qualitative shift in the predicted direction.

\begin{figure}[H]
\centering
\includegraphics[scale=0.3]{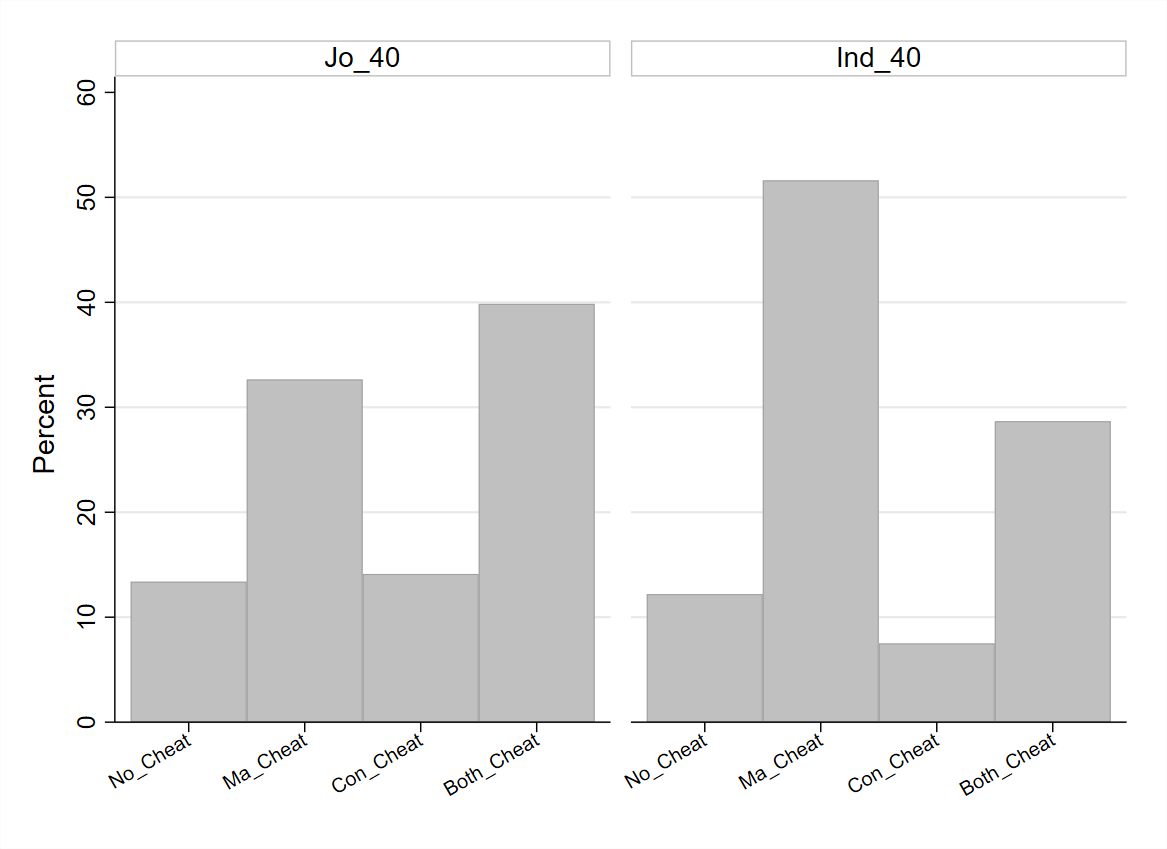}
\caption{Distributions of Team Type for Additional Treatments
\label{fig:robustcheatdistribution}}
\end{figure}

We again use a Multinomial Logit regression for the latter half of the experiment to test for treatment differences.  Figure \ref{fig:robustcheatreg} presents the predicted probabilities for combinations of cheating decisions within teams for the two additional treatments\footnote{Detailed regression results are provided in the appendix.}. The predicted probability for only the manager to cheat is significantly lower in the \textit{Ind\_40} treatment than in the \textit{Jo\_40} treatment ($p=0.005$) while there are no significant differences of the probabilities for the other categories. 

\begin{figure}[H]
\centering
\includegraphics[scale=0.3]{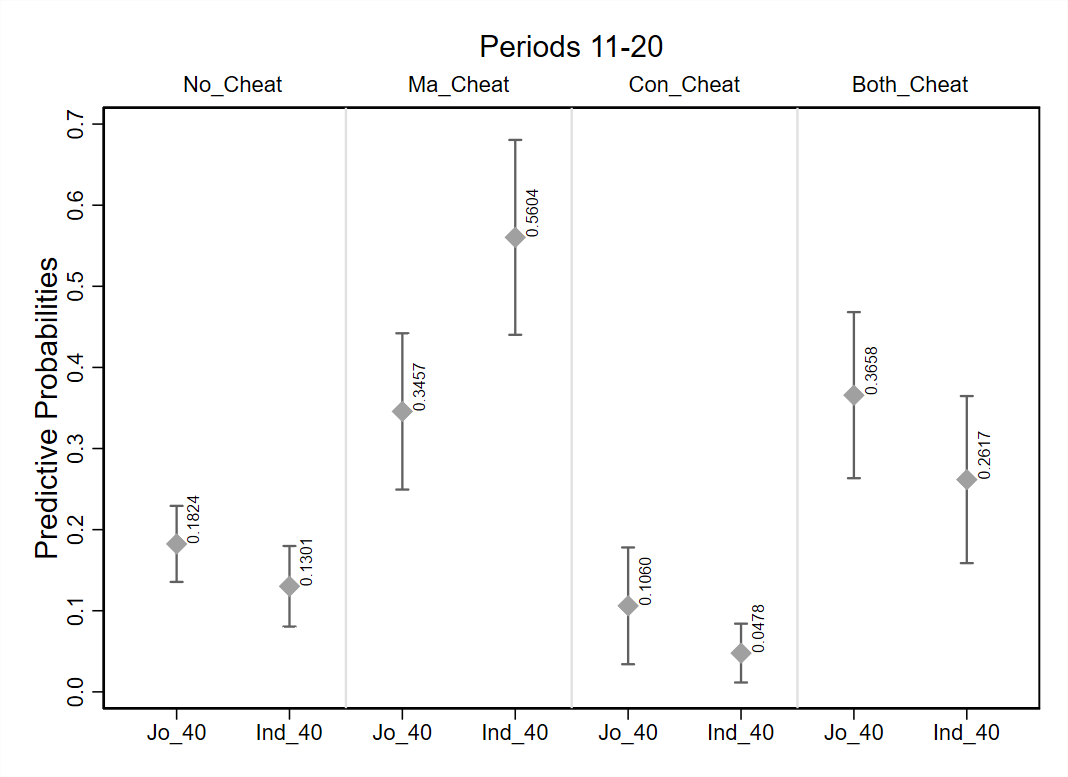}
\caption{Predictive Probabilities of Team Types for Additional Treatments
\label{fig:robustcheatreg}}
\end{figure}

Formal tests of the predicted cheating probability from the Multinomial Logit confirm the descriptive results from above. The small difference in behaviour does not translate to substantive differences of the predicted cheating rates. Initially, cheating rates tend to be slightly higher in the joint liability treatment. However, the difference for the first half of the experiment is not significant ($p=0.141$). The estimated albeit insignificant gap disappears in the second half of the experiment ($p=0.665$).

Effort patterns in the \textit{Jo\_40} and \textit{Ind\_40} treatments remain consistent with earlier findings, featuring over-dissipation, especially in the initial quarter. Contestants again exhibit irrationally high effort levels when faced with opponents exerting excessive effort. A random-effects regression on over-dissipation again does not find a treatment effect. \footnote{Detailed regression results are provided in the appendix.}

\section{\label{sec:Conclusion4}Conclusion}
In competitive settings, contestants often interact with non-contestants and work as a team. The problem of non-contesting team members to initiate cheating on behalf of the contestants without facing consequences is highly relevant in the real world but has received little attention in the literature. We investigated the underlying agency problem within a game-theoretical model. We were particularly interested in the effect of extending the liability for detected cheating to non-contestants. Our theoretical model suggests that joint liability can reduce cheating if deterrence is strong but might backfire when deterrence is weak. 

Given the secretive nature of cheating, empirical tests of our theoretical findings in the field are difficult. Hence, we designed a laboratory experiment to test the predictions of our model. While the laboratory lacks many important aspects of real-world cheating decisions, it is a good testbed for the incentive mechanism driving our theoretical results. 

The experimental results partly confirm our theoretical predictions. Under high fines, joint liability effectively deters cheating, as anticipated. However, contrary to theory, we do not find evidence of joint liability backfiring when fines are low. 

Moreover, joint liability under high fines has the additional benefit of potentially being able to reduce over-dissipation of rents in contests with cheating. Since over-exertion of effort is more pronounced when both teams engage in cheating, the ability of joint liability to reduce cheating can also limit excessive effort.

Our findings indicate that policies extending liability to non-contestants that have influence on their team's cheating decisions are not risky even if policy makers are not sure if the deterrence level is above the theoretical threshold that rules out backfiring. In sectors where cheating is commonplace and costly to detect, the introduction of shared liability can serve as a deterrent, curbing unethical behaviour without necessarily escalating sanctions to excessively high levels. By also holding those who influence but do not directly participate in the contest responsible, policy makers can potentially reduce the prevalence of cheating and at the same time distribute punishment more fairly. 

\clearpage

\appendix
\section{Proof of the Theorem}
\begin{proof}
Using Definition \ref{eq:crith} and substituting for $pf$ shows that  
\begin{equation*}
    pf>\frac{r(\delta-1)}{2(1+\delta)}\ \textrm{and}\ \eta=1 \Rightarrow \hat{\mu}<
    \frac{1}{2}.
\end{equation*}
Hence, according to Proposition \ref{propeqm} we have:
    \begin{equation*}
        pf>\frac{r(\delta-1)}{2(1+\delta)} \Rightarrow \bar{p}_h (1,pf)=\bar{p}_l (1,pf)=\frac{1}{2}.
    \end{equation*}
 Next we show that for average cheating probabilities above (below) $1/2$ we require fines below (above) $(\delta-1)/(2(1+\delta))$, which proves parts 1.) and 3.) of the Theorem.  
 From Propositions \ref{propeqm} and \ref{propeqm2}, we know that $\bar{p}_l>1/2$ requires $\eta>\hat{\eta}$ and $\hat{\mu}>1/2$, which implies:
\begin{align*}
     \eta&>1-\frac{(\delta-1)(1-r)}{(1+\delta)2fp} \\
     \frac{1}{2}&>\frac{r\left( \left(\frac{\delta}{1+\delta}\right)^{2}-\frac{1}{4}\right)-pf\eta}{r \left( \frac{1+\delta^{2}}{\left(1+\delta\right)^{2}}-\frac{1}{2}\right)}
\end{align*}
Combining the inequalities yields the necessary and sufficient condition:
\begin{equation*}
    pf<\frac{\delta-1}{2(1+\delta)}.
\end{equation*}    
Checking Propositions \ref{propeqm} and \ref{propeqm2} again, we see that $\bar{p}_h<1/2$ requires $\eta<\hat{\eta}$ and $\hat{\mu}<1/2$, which implies:
\begin{align*}
     \eta&<1-\frac{(\delta-1)(1-r)}{(1+\delta)2fp} \\
     \frac{1}{2}&<\frac{r\left( \left(\frac{\delta}{1+\delta}\right)^{2}-\frac{1}{4}\right)-pf\eta}{r \left( \frac{1+\delta^{2}}{\left(1+\delta\right)^{2}}-\frac{1}{2}\right)}
\end{align*}
Combining the inequalities yields the necessary and sufficient condition:
\begin{equation*}
    pf>\frac{\delta-1}{2(1+\delta)}.
\end{equation*}    
Together, the the two conditions prove 1.) and 3.) of the Theorem. It remains to establish that the cheating probability is always equal to $1/2$ if
\begin{equation*}
    pf=\frac{\delta-1}{2(1+\delta)}.
\end{equation*}    
This implies
\begin{align*}
     \hat{\eta}&=r \\
     \hat{\mu}&=\frac{r(1+3\delta)-2\eta(1+\delta)}{2r(\delta-1)}
\end{align*}
Combining the two equations implies the following relation:
\begin{align*}
\hat{\mu}& <\frac{1}{2} \ \text{if \ensuremath{\eta>\hat{\eta}}},\\
\hat{\mu}&=\frac{1}{2} \ \text{if \ensuremath{\eta=\hat{\eta}}},\\
\hat{\mu}&>\frac{1}{2} \ \text{if \ensuremath{\eta<\hat{\eta}}}.
\end{align*}
From Propositions \ref{propeqm} and \ref{propeqm2} we see that the equilibria in all three cases yield $\bar{p}=1/2.$
\end{proof}

\section{Random-effects multinomial regressions on team cheating type for main treatments}
\begin{table}[h]
\centering
\begin{tabular}{c ccc ccc}
\hline
\hline
\multirow{2}{*}{} & \multicolumn{3}{c}{\begin{tabular}{@{}c@{}}Model A\\First Half\end{tabular}} & \multicolumn{3}{c}{\begin{tabular}{@{}c@{}}Model B\\Second Half\end{tabular}} \\
\cline{2-7}
 & \begin{tabular}{@{}c@{}c@{}c@{}}(1)\\Ma\_Cheat\\versus\\No\_Cheat \end{tabular} & \begin{tabular}{@{}c@{}c@{}c@{}}(2)\\Con\_Cheat\\versus\\No\_Cheat \end{tabular} & \begin{tabular}{@{}c@{}c@{}c@{}}(3)\\Both\_Cheat\\versus\\No\_Cheat \end{tabular} & \begin{tabular}{@{}c@{}c@{}c@{}}(1)\\Ma\_Cheat\\versus\\No\_Cheat \end{tabular} & \begin{tabular}{@{}c@{}c@{}c@{}}(2)\\Con\_Cheat\\versus\\No\_Cheat \end{tabular} & \begin{tabular}{@{}c@{}c@{}c@{}}(3)\\Both\_Cheat\\versus\\No\_Cheat \end{tabular} \\
\hline
{\begin{tabular}{@{}c@{}}Treatment\\(Base:\textit{Jo\_H})\end{tabular}} & & & & & & \\ 

\textit{Jo\_L} & {\begin{tabular}{@{}c@{}}0.1930\\(0.5746)\end{tabular}} & {\begin{tabular}{@{}c@{}}0.3997\\(0.5914)\end{tabular}} & {\begin{tabular}{@{}c@{}}0.5983\\(0.6836)\end{tabular}} &  {\begin{tabular}{@{}c@{}}0.8726\\(0.6494)\end{tabular}} & {\begin{tabular}{@{}c@{}}1.9571**\\(0.8385)\end{tabular}} & {\begin{tabular}{@{}c@{}}1.5352\\(1.0074)\end{tabular}} \\

\textit{Ind\_H} & {\begin{tabular}{@{}c@{}}1.3678**\\(0.6352)\end{tabular}} & {\begin{tabular}{@{}c@{}}-0.9744*\\(0.5678)\end{tabular}} & {\begin{tabular}{@{}c@{}}0.5458\\(0.6836)\end{tabular}} & {\begin{tabular}{@{}c@{}}2.3155***\\(0.7346)\end{tabular}} & {\begin{tabular}{@{}c@{}}-0.6302\\(0.8874)\end{tabular}} & {\begin{tabular}{@{}c@{}}1.3672\\(1.0161)\end{tabular}} \\

\textit{Ind\_L} & {\begin{tabular}{@{}c@{}}1.3238*\\(0.7343)\end{tabular}} & {\begin{tabular}{@{}c@{}}-1.1839*\\(0.6807)\end{tabular}} & {\begin{tabular}{@{}c@{}}0.7382\\(0.7841)\end{tabular}} & {\begin{tabular}{@{}c@{}}2.4670***\\(0.8286)\end{tabular}} & {\begin{tabular}{@{}c@{}}-0.4634\\(1.0940)\end{tabular}} & {\begin{tabular}{@{}c@{}}1.6547\\(1.0225)\end{tabular}} \\

Constant & {\begin{tabular}{@{}c@{}}-0.0698\\(0.5323)\end{tabular}}  & {\begin{tabular}{@{}c@{}}-0.0840\\(0.4543)\end{tabular}} & {\begin{tabular}{@{}c@{}}0.4728\\(0.5607)\end{tabular}} & {\begin{tabular}{@{}c@{}}-0.7652\\(0.5015)\end{tabular}} & {\begin{tabular}{@{}c@{}}-2.0001***\\(0.7013)\end{tabular}} & {\begin{tabular}{@{}c@{}}-0.6008\\(0.7510)\end{tabular}} \\

\textit{Var(u2)} & 1.7875 & 1.7875 & 1.7875 & 3.4111 & 3.4111 & 3.4111\\
\textit{Var(u3)} & 1.9041 & 1.9041 & 1.9041 & 6.0574 & 6.0574 & 6.0574\\
\textit{Var(u4)} & 2.1986 & 2.1986 & 2.1986 & 9.5636 & 9.5636 & 9.5636 \\
\textit{N} & 1240 & 1240 & 1240 & 1240 & 1240 & 1240 \\
\hline
\multicolumn{7}{l}{{\small{}Note: Robust standard errors in parentheses, clustering on the group level}}\\
\multicolumn{7}{l}{{\small{}{*}{*}{*} Sig. at the 1 percent level, {*}{*}Sig. at the 5 percent level, {*}Sig. at the 10 percent level}}\\
\end{tabular}
\end{table}

\section{Random-effects multinomial regressions on team cheating type for additional treatments}
\begin{table}[H]
\centering
\begin{tabular}{c ccc ccc}
\hline
\hline
\multirow{2}{*}{} & \multicolumn{3}{c}{\begin{tabular}{@{}c@{}}Model C\\(First Half)\end{tabular}} & \multicolumn{3}{c}{\begin{tabular}{@{}c@{}}Model D\\(Second Half)\end{tabular}} \\
\cline{2-7}
 & \begin{tabular}{@{}c@{}c@{}c@{}}(1)\\Ma\_Cheat\\versus\\No\_Cheat\end{tabular} & \begin{tabular}{@{}c@{}c@{}c@{}}(2)\\Con\_Cheat\\versus\\No\_Cheat\end{tabular} & \begin{tabular}{@{}c@{}c@{}c@{}}(3)\\Both\_Cheat\\versus\\No\_Cheat\end{tabular} & \begin{tabular}{@{}c@{}c@{}c@{}}(1)\\Ma\_Cheat\\versus\\No\_Cheat\end{tabular} & \begin{tabular}{@{}c@{}c@{}c@{}}(2)\\Con\_Cheat\\versus\\No\_Cheat\end{tabular} & \begin{tabular}{@{}c@{}c@{}c@{}}(3)\\Both\_Cheat\\versus\\No\_Cheat\end{tabular} \\
\hline
{\begin{tabular}{@{}c@{}}Treatment\\(Base:\textit{Jo\_40})\end{tabular}} & & & & & & \\ 

\textit{Ind\_40} & 
\begin{tabular}{@{}c@{}}0.4539\\(0.5994)\end{tabular} & 
\begin{tabular}{@{}c@{}}-0.7614\\(0.6465)\end{tabular} & 
\begin{tabular}{@{}c@{}}-0.8245\\(0.6942)\end{tabular} & 
\begin{tabular}{@{}c@{}}1.5892**\\(0.7604)\end{tabular} &
\begin{tabular}{@{}c@{}}-0.7631\\(0.9396)\end{tabular} & 
\begin{tabular}{@{}c@{}}-0.0955\\(0.9757)\end{tabular} \\

Constant & 
\begin{tabular}{@{}c@{}}0.9877**\\(0.3941)\end{tabular} &
\begin{tabular}{@{}c@{}}0.2469\\(0.4378)\end{tabular} &
\begin{tabular}{@{}c@{}}1.3652***\\(0.4605)\end{tabular} &
\begin{tabular}{@{}c@{}}0.9281**\\(0.4515)\end{tabular} &
\begin{tabular}{@{}c@{}}-1.0495\\(0.6427)\end{tabular} &
\begin{tabular}{@{}c@{}}0.9709\\(0.6950)\end{tabular} \\

\textit{Var(u2)} & 3.3295 & 3.3295 & 3.3295 & 5.0463 & 5.0463 & 5.0463 \\
\textit{Var(u3)} & 1.6177 & 1.6177 & 1.6177 & 4.3836 & 4.3836 & 4.3836 \\
\textit{Var(u4)} & 2.5428 & 2.5428 & 2.5428 & 8.9900 & 8.9900 & 8.9900 \\

\textit{N} & 680 & 680 & 680 & 680 & 680 & 680 \\
\hline
\multicolumn{7}{l}{{\small{}Note: Robust standard errors in parentheses, clustering on the group level}}\\
\multicolumn{7}{l}{{\small{}{*}{*}{*} Sig. at the 1 percent level, {*}{*} Sig. at the 5 percent level, {*} Sig. at the 10 percent level}}\\
\end{tabular}
\end{table}

\section{Observed effort versus best response effort for additional treatments}

\begin{figure}[H]
\centering
\includegraphics[scale=0.23]{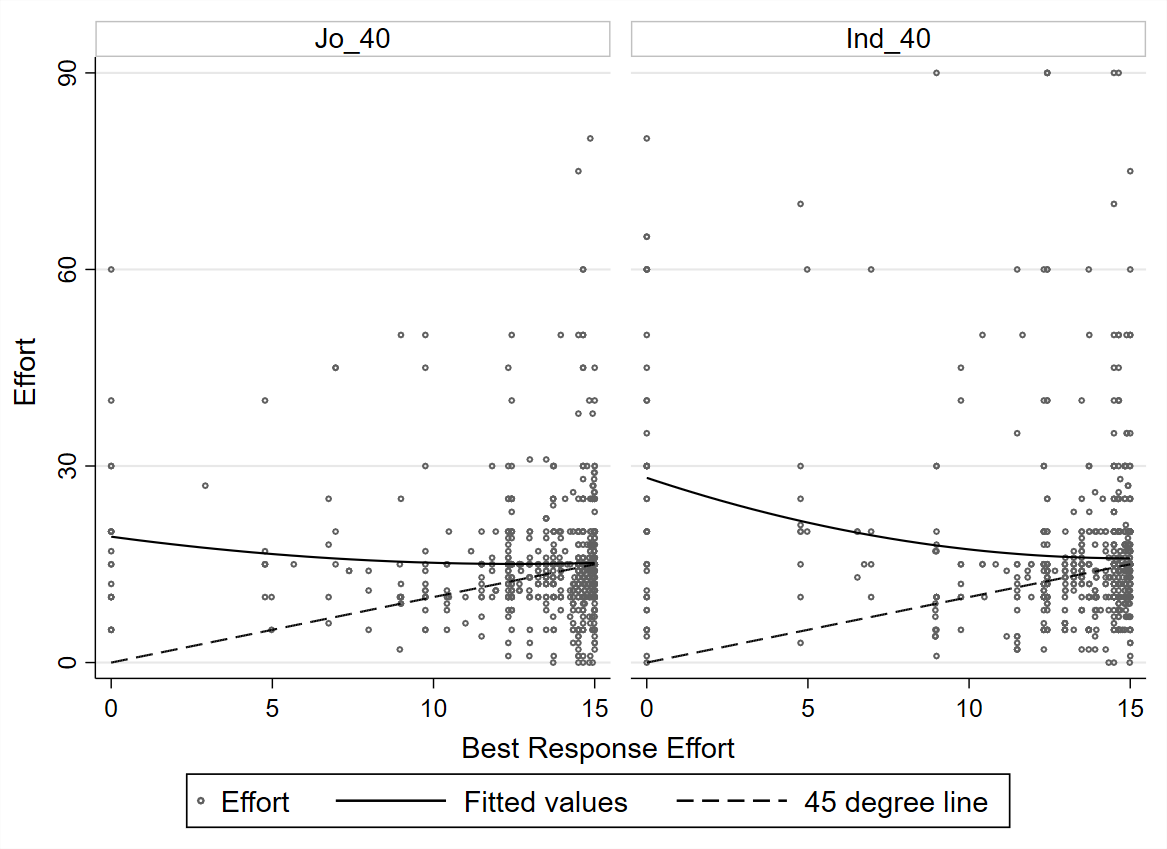}
\caption{Observed Effort versus Best Response Effort by Treatment \label{fig:Effort-versus-best}}
\end{figure}

\section{Random-effects regression on over-dissipation for additional treatments}
\begin{table}[H]
\centering
\begin{tabular}{ m{4.5cm} >{\centering\arraybackslash}m{2.5cm} >{\centering\arraybackslash}m{2.5cm} >{\centering\arraybackslash}m{2.5cm} }
\hline
\hline
& \begin{tabular}{@{}c@{}}Full\\Model\end{tabular} & \begin{tabular}{@{}c@{}}First-Half\\Model\end{tabular} & \begin{tabular}{@{}c@{}}Second-Half\\Model\end{tabular}\\
\hline
Treatment (Base:\textit{Jo\_40}) & & & \\ 
\textit{Ind\_40} & \begin{tabular}{@{}c@{}}1.7032\\(2.0346)\end{tabular} & \begin{tabular}{@{}c@{}}3.7571\\(2.8654)\end{tabular} & \begin{tabular}{@{}c@{}}-0.0898\\(1.6135)\end{tabular} \\[6pt]

Number of Cheating Teams (Base: 0) & & & \\ 
\textit{One Team Cheats} & \begin{tabular}{@{}c@{}}2.4104***\\(0.8024)\end{tabular} & \begin{tabular}{@{}c@{}}2.7051**\\(1.1251)\end{tabular} & \begin{tabular}{@{}c@{}}1.5029\\(1.0192)\end{tabular} \\[6pt]

\textit{Both Teams Cheat} & \begin{tabular}{@{}c@{}}1.5346\\(1.0282)\end{tabular} & \begin{tabular}{@{}c@{}}3.7484***\\(1.4299)\end{tabular} & \begin{tabular}{@{}c@{}}-0.6864\\(1.1628)\end{tabular} \\[6pt]

\multicolumn{1}{l}{Quarter (Base: Q1; Q3)} & & & \\
\textit{Q2} & \begin{tabular}{@{}c@{}}-4.3740***\\(1.3324)\end{tabular} & \begin{tabular}{@{}c@{}}-4.3306***\\(1.3371)\end{tabular} & \\[3pt]
\textit{Q3} & \begin{tabular}{@{}c@{}}-4.7501***\\(1.5072)\end{tabular} & & \\[3pt]
\textit{Q4} & \begin{tabular}{@{}c@{}}-5.7913***\\(1.6825)\end{tabular} & & \begin{tabular}{@{}c@{}}-0.8838\\(0.6336)\end{tabular} \\[6pt]

Intercept & \begin{tabular}{@{}c@{}}3.0650**\\(1.4621)\end{tabular} & \begin{tabular}{@{}c@{}}1.0990\\(1.5829)\end{tabular} & \begin{tabular}{@{}c@{}}0.4616\\(1.1981)\end{tabular} \\
\hline
\multicolumn{4}{l}{{\small{}{*}{*}{*} Sig. at the 1 percent level, {*}{*}Sig. at the 5 percent level, {*}Sig. at the 10 percent level}}\\
\end{tabular}
\caption{Random-effects Regression on Deviation in Contribution Effort, Clustering on the Group Level}
\label{tab:random-effect-devconteff}
\end{table}

\newpage 

\bibliographystyle{chicago}
\bibliography{masterbibfile}

\end{document}